\documentclass[11pt]{article}
\usepackage{amsthm}
\usepackage{amsmath}
\usepackage{amssymb}
\usepackage{amsfonts}

\usepackage{braket}
\usepackage{comment}
\usepackage{mathtools}
\usepackage{graphics}
\usepackage{algorithm}
\usepackage{color}
\usepackage{comment}
\usepackage{hyperref}

\usepackage{lineno}
\linenumbers
\modulolinenumbers[100]
\pagewiselinenumbers

\newcommand{\search}{{\sf Search}}
\newcommand{\searchfast}{{\sf Search^{\star}}}
\newcommand{\ascend}{{\sf Ascend}}

\newcommand{\scan}{{\sf Scan}}
\newcommand{\update}{{\sf Update}}

\def\phi{\varphi}
\def\odd{\mathrm{odd}}

\def\K{\mathcal{K}}
\def\X{\mathcal{X}}
\def\P{\mathcal{P}}
\def\G{\mathcal{G}}
\def\S{\mathcal{S}}
\def\T{\mathcal{T}}
\def\A{\mathcal{A}}
\def\B{\mathcal{B}}
\def\C{\mathcal{C}}
\def\Zp{\mathbf{Z}_{+}}

\newcounter{rnc}

\newcommand{\rn}[1]{\setcounter{rnc}{#1}\roman{rnc}}

\newtheorem{theorem}{Theorem}[section]
\newtheorem{lemma}[theorem]{Lemma}
\newtheorem{claim}[theorem]{Claim}
\newtheorem{corollary}[theorem]{Corollary}
\newtheorem{proposition}[theorem]{Proposition}

\newcommand{\seq}{{\sf seq}}
\newcommand{\sym}{{\sf sym}}

\newcommand{\minimal}{{\sf min}}

\newcommand{\reconst}{{\sf Reconst}}
\newcommand{\grow}{{\sf Grow}}
\newcommand{\find}{{\sf Find}}
\newcommand{\link}{{\sf Link}}

\def\F{\mathcal{F}}
\def\val{\mathrm{val}}
\def\OPT{\mathrm{opt}_\mathbf{R}}
\def\OPTint{\mathrm{opt}_\mathbf{Z}}
\def\odd{\mathrm{odd}}

\def\K{\mathcal{K}}
\def\X{\mathcal{X}}
\def\P{\mathcal{P}}
\def\G{\mathcal{G}}

\def\V{\mathcal{V}}
\def\Y{\mathcal{Y}}

\title{Finding Maximum Edge-Disjoint Paths \\ 
Between Multiple Terminals\thanks{A preliminary version has appeared in Proceedings of the 31st Annual ACM-SIAM
Symposium on Discrete Algorithms (SODA 2020), pp.~1933--1944.}
}
\author{Satoru Iwata
\thanks{Department of Mathematical Informatics, University of Tokyo, Tokyo 113-8656, Japan, and Institute for Chemical Reaction Design and Discovery, Hokkaido University, Sapporo, Hokkaido, 001-0021, Japan.
Supported by Grants-in-Aid for Scientific Research, 17H01699 from JSPS and 20H05965 from MEXT. Also supported by CREST JPMJCR14D2 and ERATO JPMJER1903 from JST.}
\and 
Yu Yokoi
\thanks{National Institute of Informatics, Tokyo 101-8430, Japan. Supported by JSPS KAKENHI JP18K18004, JST PRESTO JPMJPR212B, and JST CREST JPMJCR14D2.}
}

\begin{document}
\maketitle
\begin{abstract}
Let $G=(V,E)$ be a multigraph with a set $T\subseteq V$ of terminals. A path in $G$ is called a $T$-path if 
its ends are distinct vertices in $T$ and no internal vertices belong to $T$. In 1978, Mader showed 
a characterization of the maximum number of edge-disjoint $T$-paths. 

In this paper, we provide a combinatorial, deterministic algorithm for finding the maximum number of 
edge-disjoint $T$-paths. The algorithm adopts an augmenting path approach. More specifically, we utilize a new concept of short augmenting walks in auxiliary labeled graphs to capture a possible augmentation of 
the number of edge-disjoint $T$-paths. To design a search procedure for a short augmenting walk, we introduce 
blossoms analogously to the matching algorithm of Edmonds (1965). When the search procedure terminates without finding a short  augmenting walk, 
the algorithm provides a certificate for the optimality of the current edge-disjoint $T$-paths. 
From this certificate, one can obtain the Edmonds--Gallai type decomposition introduced by Seb\H{o} and Szeg\H{o} (2004).  
The algorithm runs in $O(|E|^2)$ time, which is much faster than 
the best known deterministic algorithm based on a reduction to linear matroid parity. 

We also present a strongly polynomial algorithm for the maximum integer free multiflow problem, 
which asks for a nonnegative integer combination of $T$-paths maximizing the sum of the coefficients 
subject to capacity constraints on the edges. 
\end{abstract}

\section{Introduction}
Let $G=(V,E)$ be a multigraph without selfloops. For a specified set $T\subseteq V$ of terminals,  
a path in $G$ is called a $T$-path if its ends are distinct vertices in $T$ and no internal vertices belong to $T$. 
In 1978, Mader~\cite{MaderA} showed a characterization of the maximum number of edge-disjoint $T$-paths. 
The theorem naturally extended a previously known min-max theorem on the inner Eulerian case due to 
Cherkassky~\cite{Cherkassky77} and Lov\'asz~\cite{Lovasz76}. Unlike this preceding result, the original proof was not constructive. 

Subsequently, Mader~\cite{MaderH} extended his theorem to the problem of maximum number of openly disjoint $T$-paths. 
The result contains as its special case Gallai's min-max theorem on maximum number of vertex-disjoint $T$-paths \cite{Gallai61}, which is equivalent to the Tutte-Berge formula on maximum matching \cite{Berge58,Tutte47}.
Lov\'asz~\cite{Lovasz80} then introduced an equivalent variant, called disjoint $\A$-paths, where $\A$ is 
a given partition of the terminals, to provide 
an alternative proof via the matroid matching theorem. See also \cite{TY16} for a minor correction. 
Schrijver~\cite{Schr01} provided a short alternative proof for Mader's theorem on disjoint $\A$-paths based on Gallai's min-max theorem.
The proof was again nonconstructive, and it did not lead to an efficient algorithm. Analogously to the Edmonds--Gallai decomposition for maximum matching, 
Seb\H{o} and Szeg\H{o} \cite{SS04} introduced a canonical decomposition that captures all the disjoint $\A$-paths.

Schrijver~\cite{Schr03} described a reduction of the disjoint $\A$-paths problem to the linear matroid parity problem. 
Consequently, one can use efficient linear matroid parity algorithms \cite{CLL14,GS86,Orlin08,OV92} for finding 
the maximum number of disjoint $\A$-paths (or openly disjoint $T$-paths). The current best running time bound 
is $O(n^\omega)$, where $n$ is the number of vertices and $\omega$ is the exponent of the fast matrix multiplications. 
This bound is achieved by the randomized algebraic algorithm of Cheung, Lau, and Leung~\cite{CLL14}. The best deterministic 
running time bound due to Gabow and Stallmann \cite{GS86} is $O(mn^\omega)$, where $m$ is the number of edges. 
Without using the reduction to linear matroid parity, Chudnovsky, Cunningham, and Geelen~\cite{CCG08} devised 
a combinatorial algorithm that runs in $O(n^5)$ time. 
When applying these methods to the edge-disjoint $T$-paths problem, one has to deal with the line graph of $G$. 
Thus the best known randomized and deterministic running time bounds for the maximum edge-disjoint $T$-paths
are $O(|E|^\omega)$ and $O(|E|^{\omega+2})$, respectively. 

Another approach to finding maximum edge-disjoint $T$-paths is based on linear programming. Keijsper, Pendavingh, and Stougie~\cite{KPS06} provided a dual pair of linear programs whose optimal value coincides with the maximum number of edge-disjoint $T$-paths. Giving an efficient separation procedure for this linear program, they showed that one can find maximum edge-disjoint $T$-paths in polynomial time via the ellipsoid method. 

In this paper, we provide a combinatorial, deterministic algorithm for finding the maximum number of edge-disjoint 
$T$-paths. The algorithm adopts an augmenting path approach. More specifically, we introduce a novel concept of 
augmenting walks in auxiliary labeled graphs to capture a possible augmentation of the number of edge-disjoint $T$-paths. 
To design a search procedure for an augmenting walk, we introduce blossoms analogously to the matching algorithm 
of Edmonds~\cite{Edmonds65}, although the present problem is neither a special case nor a generalization of the matching problem. 
When the search procedure terminates without finding an augmenting walk, the algorithm provides a certificate for the optimality of the current edge-disjoint $T$-paths.
Thus the correctness argument of the algorithm serves as an alternative direct proof of Mader's theorem on edge-disjoint $T$-paths. 
In addition, the optimality certificate obtained by the algorithm coincides with the Edmonds--Gallai type decomposition for edge-disjoint $T$-paths introduced by 
Seb\H{o} and Szeg\H{o} \cite{SS04}. 
The algorithm runs in $O(|E|^2)$ time. This is definitely faster than the above mentioned algorithms.

A preliminary version of this paper presented an $O(|V|\cdot |E|^2)$ algorithm for finding maximum edge-disjoint $T$-paths \cite{IwataYokoi20}. 
In that version, 
however, performing an augmentation step was quite complicated. In fact, it took $O(|V|\cdot |E|)$ time for an augmentation, and the search procedure also took $O(|V|\cdot |E|)$ time, 
which led to the 
$O(|V|\cdot|E|^2)$ running time bound in total.  
Subsequently, Hummel~\cite{Hummel20} improved the algorithm to run in $O((|V|^2+|E|)\cdot |E|)$ time.

In the present paper, we focus on short augmenting walks, which mean augmenting walks without certain structures called shortcuts. It will be shown that an augmentation along a short augmenting walk takes $O(|E|)$ time and results in a larger number of edge-disjoint $T$-paths. We also devise a procedure to find a short augmenting walk in $O(|E|)$ time. Consequently, one can obtain 
$k$ edge-disjoint $T$-paths in $O(k\,|E|)$ time. Since the number of edge-disjoint $T$-paths is at most $|E|$, the overall running time is $O(|E|^2)$.

A natural generalization of the present setting is to think of finding maximum edge-disjoint $T$-paths of 
minimum total cost, where the cost is defined to be the sum of the costs of the included edges. 
Karzanov~\cite{Karzanov97} gave a min-max theorem and described a combinatorial algorithm for this problem. 
The detailed proof of correctness was given in a technical report of more than 60 pages \cite{Karzanov93}.
Hirai and Pap~\cite{HiraiPap14} dealt with a weighted maximization of edge-disjoint $T$-paths, 
where the weight is given by a metric on the terminal set $T$. They clarified that this problem with edge costs 
can be solved in polynomial time if the weight is given by a tree metric and that it is NP-hard otherwise. 
Mader's edge-disjoint $T$-paths problem corresponds to the case with a tree metric that comes from a star. 
They adopted a novel polyhedral approach to prove a min-max theorem that extends Mader's theorem on 
edge-disjoint $T$-paths. Their algorithm, however, depends on the ellipsoid method. 
Our algorithm may serve as a prototype of possible combinatorial algorithms for this generalization. 

Another natural generalization is a capacitated version, which is called the integer free multiflow problem.
For the fractional relaxation of this problem, one can find an optimal solution, which is half-integral, 
in strongly polynomial time with the aid of an algorithm by Ibaraki, Karzanov, and Nagamochi~\cite{IKN98}. 
Rounding down this half-integral solution, we obtain an integral multiflow, which serves as an initial feasible solution. 
We can transform this feasible solution to an optimal solution by repeatedly applying our search and augmentation procedures developed for maximum edge-disjoint $T$-paths.
The number of augmentations 
is bounded by a polynomial in the size of the graph. Thus we obtain a strongly polynomial algorithm for the 
integer free multiflow problem. 

\smallskip
The rest of this paper is organized as follows.
Section~\ref{sec:Mader} describes the statement of Mader's theorem on
edge-disjoint $T$-paths. Before providing a precise description of our algorithm, we first explain motivations of our algorithm design in Section~\ref{sec:outline}. In Section~\ref{sec:augmentation}, we formally introduce the notion of short augmenting walks in auxiliary labeled graphs, and provide a procedure to increase the number of edge-disjoint $T$-paths
using a short augmenting walk.
Section~\ref{sec:search} presents a procedure to find a short augmenting walk. The detailed implementation and time complexity are discussed in Section~\ref{sec:implementation}.
Sections~\ref{sec:EG} is devoted to the connection to the Edmonds--Gallai type decomposition. Finally, in Section~\ref{sec:IFMF},  we extend our algorithm to solve the maximum integer free multiflow problem.

\section{Mader's Theorem}\label{sec:Mader}
For a multigraph $G=(V,E)$ and a set $T\subseteq V$ of terminals, a collection $\X=\{X_s\}_{s\in T}$ of mutually disjoint sets $X_s\subseteq V$  is called a $T$-subpartition 
if $X_s\cap T=\{s\}$ holds for each $s\in T$. 
For any $X\subseteq V$, we denote by $\delta(X)$ the set of edges 
in $E$ between $X$ and $V\setminus X$. We also denote $d(X)\coloneqq |\delta(X)|$.

Since each $T$-path between $s,t\in T$ contains at least one edge in 
$\delta(X_s)$ and at least one edge in $\delta(X_t)$, the number of 
edge-disjoint $T$-paths is at most $\frac{1}{2}\sum_{s\in T} d(X_s)$. 
This is not a tight bound. A more detailed analysis, however, 
leads to a tighter upper bound as follows.

For a $T$-subpartition $\X$, let $G\setminus \X$ denote 
the graph obtained from $G$ by deleting all the vertices 
in $\bigcup_{s\in T} X_s$ and incident edges. 
A connected component of $G\setminus\X$ is said to be {\em odd} 
if its vertex set $K$ has odd $d(K)$ and {\em even} otherwise. Then $\odd(G\setminus \X)$ 
denotes the number of odd components in $G\setminus\X$. 

\begin{lemma}
\label{lem:Mader}
For a multigraph $G=(V,E)$ and a set $T\subseteq V$, the number of edge-disjoint $T$-paths is at most 
$$\kappa(\X)\coloneqq \frac{1}{2}\left[\sum_{s\in T}d(X_s)-\odd(G\setminus\X)\right]$$
for any $T$-subpartition $\X$.  
\end{lemma}
\begin{proof}
Suppose that there are $k$ edge-disjoint $T$-paths in $G$. 
Let $\K$ be the collection of vertex sets of connected components of 
$G\setminus\X$. A $T$-path contains an edge between two distinct 
components of $\X$ or passes through a member $K$ of $\K$. In the latter 
case, the $T$-path must contain two edges in $\delta(K)$. Therefore, 
we have the inequality
$k\leq d(\X)+\sum_{K\in\K}\left\lfloor\frac{d(K)}{2}\right\rfloor,$
where $d(\X)$ is the number of edges between distinct components of $\X$. 
Since 
$d(\X)=\frac{1}{2}\left[\sum_{s\in T}d(X_s)-\sum_{K\in\K}d(K)\right],$
this implies  
$k\leq\frac{1}{2}\left[\sum_{s\in T}d(X_s)-\sum_{K\in\K}\left(d(K)
-2\left\lfloor\frac{d(K)}{2}\right\rfloor\right)\right],$
and the right-hand side equals $\kappa(\X)$. 
\end{proof}

Mader's edge-disjoint $T$-paths theorem asserts that this upper bound is tight. 

\begin{theorem}[Mader~\cite{MaderA}]
\label{th:Mader}
The maximum number of edge-disjoint $T$-paths equals  
the minimum of $\kappa(\X)$ among all the $T$-subpartitions $\X$. 
\end{theorem}

\section{Motivation of the Algorithm Design} \label{sec:outline}
This section describes high level ideas in the design of our algorithm 
for finding the maximum number of edge-disjoint $T$-paths. 

The algorithm iteratively increases the number of edge-disjoint $T$-paths. 
Given a family $\P$ of edge-disjoint $T$-paths, let $Q$ be a walk such that both 
ends are in $T$ and the inner vertices are in $V\setminus T$. Consider the 
symmetric difference $E(\P)\triangle E(Q)$ between $E(\P)$ and $E(Q)$, 
where $E(\P)$ denotes the set of edges in $T$-paths in $\P$ and $E(Q)$ 
the set of edges that appear in $Q$ an odd number of times. 
The resulting subgraph $H^*=(V,E(\P)\triangle E(Q))$ is inner Eulerian, i.e., 
each vertex in $V\setminus T$ has even degree. Then $H^*$ can be decomposed into 
an edge-disjoint family $\P'$ of $T$-paths and cycles. If the first and last 
edges of $Q$ are not in $E(\P)$, then the sum of the degrees of the terminals 
in $E(\P)\triangle E(Q)$ is two more than that in $E(\P)$. If in addition 
$\P'$ is free from $T$-cycles, i.e., cycles that contain exactly one terminal, 
then $\P'$ includes a family of $|\P|+1$ edge-disjoint $T$-paths. Thus a crucial 
point of the algorithm is to find an appropriate walk $Q$ so that $E(\P)\triangle E(Q)$ 
can be decomposed into $\P'$ that is free from $T$-cycles. 
 
For the sake of simplicity, assume tentatively that $Q$ is a path 
and 
$Q$ shares at most one segment (i.e., maximal subpath) with each $T$-path in $\P$. Let $P_1,\ldots,P_{\ell-1}\in\P$ 
be the $T$-paths that share a segment in $Q$ in this order. We denote the segment of $P_i$ 
shared with $Q$ by $P_i(u_i,v_i)$, where $u_i$ and $v_i$ are the first and last vertices 
of the segment, respectively. We also denote by $s_i$ and $t_i$ the ends of $P_i$ such that 
$s_i,u_i,v_i,t_i$ appear in this order along $P_i$. In addition, the first and last vertices 
of $Q$ are denoted by $t_0$ and $s_\ell$, respectively. See Figure~\ref{fig:motivation1}. Then it is easy to 
observe that $E(\P)\triangle E(Q)$ includes a $T$-cycle if $s_i=t_{i-1}$ holds for some $i=1,\ldots,\ell$.

\begin{figure}[b]
	\begin{center}		\includegraphics[width=0.46\hsize]{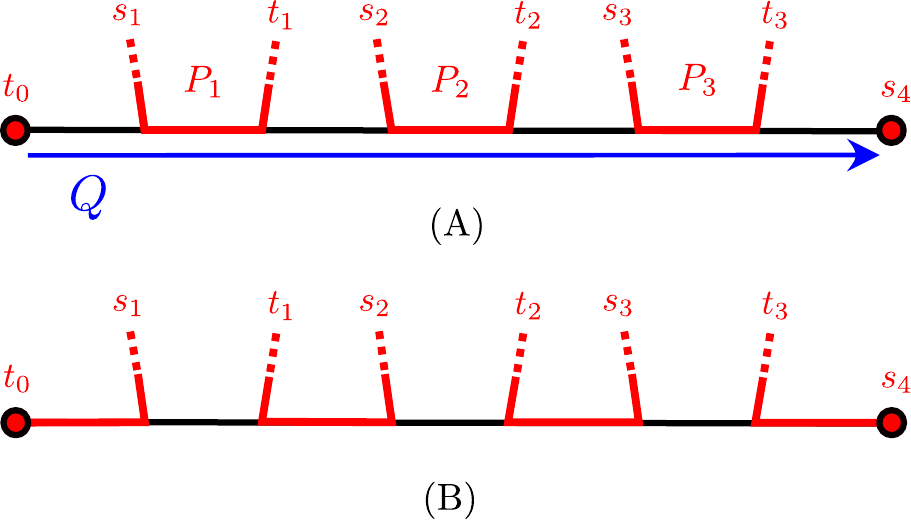}
	\end{center}
  	\vspace{-2mm}
	\caption{
		\small\baselineskip=11pt
		(A) A walk $Q$ (blue line) and the $T$-paths $P_1$, $P_2$, $P_3$ (red lines) each of which shares a subpath with $Q$. Black lines represent edges not used in any $T$-path. The walk $Q$ starts at $t_0\in T$ and ends at $s_4\in T$. Each $P_i~(i=1,2,3)$ has end-terminals $s_i$ and $t_i$. (B) The symmetric difference of $Q$ and $P_1, P_2, P_3$.
	}
	\label{fig:motivation1}
\end{figure}

In order to find an appropriate walk avoiding such 
an occurrence of a $T$-cycle,
we 
introduce an auxiliary labeled graph. Each edge in a $T$-path $P\in\P$ is labeled by $st$, where $s$ and $t$ are the ends of $P$. More precisely, for an edge $e$ with $\partial e=\{u,v\}$, 
the symbols $s$ and $t$ are assigned to the $u$-side and $v$-side of $e$, respectively, where 
$s,u,v,t$ are supposed to appear in this order along $P$. The edges in $E\setminus E(\P)$ 
are not assigned any labels. In addition, each terminal $t\in T$ is assigned $t$ as its label. 
Then the sequence of symbols that appear along a walk $Q$ is required to contain no consecutive 
appearance of a symbol. 
In the above mentioned case, this requirement on $Q$ guarantees that $s_i\neq t_{i-1}$ holds for any $i=1,\ldots,\ell$.

This requirement is not sufficient to capture all walks that lead to successful augmentations. 
A simple example depicted in Figure~\ref{fig:motivation2}\,(A) consists of three terminals $r,s,t$ and 
an inner vertex $u$. The terminals $s$ and $t$ are connected by a $T$-path through $u$. Given 
this $T$-path as the only member of $\P$, we do not find a walk without consecutive appearance of a symbol while the original graph obviously admits two edge-disjoint $T$-paths. In order to deal with such 
a situation, we attach a selfloop to each inner vertex of a $T$-path in $P\in\P$. The selfloop is assigned a label $st$, where $s$ and $t$ are the ends of $P$. (More precisely, we attach two selfloops with labels $st$ and $ts$.) Using such a selfloop in the example in Figure~\ref{fig:motivation2}\,(B), one can find a walk $Q$ from $r$ to $r$ through $u$ without consecutive 
appearance of a symbol, and the symmetric difference $E(\P)\triangle E(Q)$ can be decomposed into 
two edge-disjoint $T$-paths. 
\begin{figure}[h]
	\begin{center}
		\includegraphics[width=0.8\hsize]{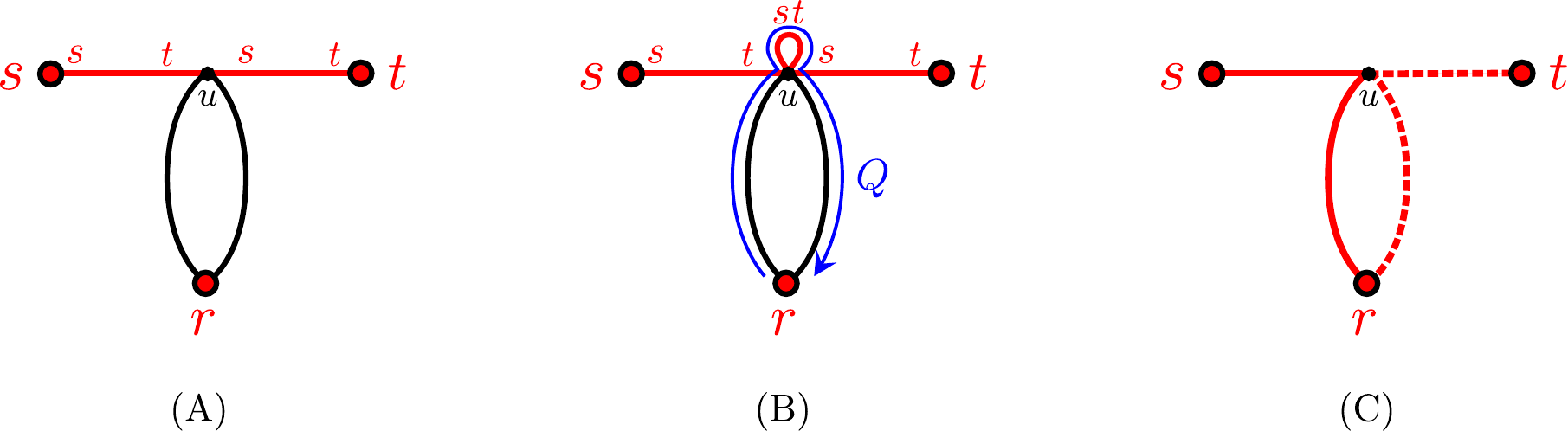}
	\end{center}
	\vspace{-4mm}
	\caption{
		\small\baselineskip=11pt
		(A) A $T$-path $P$ is represented by a red line. (B) An auxiliary graph with a selfloop at $u$. The blue line represents a walk $Q$ on which symbols $rstr$ appear in this order. (C) Two $T$-paths obtained by taking symmetric difference of $P$ and $Q$. The use of the selfloop at $u$ causes a switch of the transition at $u$. 
	}
	\label{fig:motivation2}
\end{figure}

A formal definition of the auxiliary labeled graph thus motivated will be given in Section~\ref{sec:augmentation}. 
A walk $Q$ that satisfies the above conditions is defined to be an augmenting walk there. Then we expect that the symmetric difference $E(\P)\triangle E(Q)$ necessarily leads to a larger number of 
edge-disjoint $T$-paths even without the above tentative assumption. It turns out that this 
is not always the case. Simple counterexamples will be given in Section~\ref{sec:augmentation}. 
In those examples, however, there are other augmenting walks with smaller number of $\P$-segments (i.e., maximal subwalks formed by subpaths of the $T$-paths in $\P$). 
Thus one can expect that an augmenting walk with the smallest number of $\P$-segments leads to a 
successful augmentation. This is indeed the case as will be proved in Theorem~\ref{thm:SAW1} in 
a slightly more general form. We identify a structure called shortcut that possibly prevents 
a successful augmentation. An augmenting walk with a shortcut can be transformed into another 
augmenting walk with smaller number of $\P$-segments. An augmenting walk will be called short if 
it does not admit a shortcut. The statement of Theorem~\ref{thm:SAW1} says 
that a short augmenting walk leads to a successful augmentation. 

We will then focus on a search procedure for a short augmenting walk. This is designed in Section~\ref{sec:search} 
analogously to the blossom algorithm for maximum matching \cite{Edmonds65}.  
With the aid of the techniques developed for an efficient implementation of the matching algorithm, 
one can implement the search procedure to run in $O(|E|)$ time. The details will be described in 
Section~\ref{sec:implementation}. 

\paragraph{Comparison with the previous versions.}
In the previous version \cite{IwataYokoi20}, the authors introduced the auxiliary labeled graph and augmenting walks therein.
We described a search procedure for finding an augmenting walk. Then it was also shown that a family $\P$ of 
edge-disjoint $T$-paths and an augmenting walk $Q$ can be collectively transformed into other ones with 
smaller number of $\P$-segments. Repeating this procedure, one can finally obtain an augmenting walk that 
is free from $\P$-segments.
Then it is easy to obtain a larger number of edge-disjoint $T$-paths than before. 
The repeated applications of the transformation require $O(|V|\cdot|E|)$ time. The search procedure for
an augmenting walk also required $O(|V|\cdot |E|)$ time. Hummel~\cite{Hummel20} then improved both of these procedures 
to run in $O(|V|^2+|E|)$ time. Since the number of augmentation is $O(|E|)$, the total running time of these algorithms 
are $O(|V|\cdot |E|^2)$ and $O((|V|^2+|E|)|E|)$. 

In the present paper, we show that a short augmenting walk produces a larger number of edge-disjoint $T$-paths 
in $O(|E|)$ time. We also devise a search procedure to find a short augmenting walk and improve it to run in 
$O(|E|)$ time. Thus the running time per augmentation is $O(|E|)$. Since the number of augmentations is $O(|E|)$, 
the total running time bound is $O(|E|^2)$. 

\vspace{-2mm}
\paragraph{Connection to search for a directed path.}
The search for an augmenting walk in the labeled graph 
naturally generalizes basic graph search problems. In the special case with $|T|=\{s,t\}$, there are only two symbols in our auxiliary labeled graph. Finding 
an augmenting walk in that setting is nothing but searching for a directed walk in a mixed graph, 
where all the edges in the current $s$-$t$ paths 
are directed from $t$ to $s$ and the other edges 
are undirected. 
The general setting of the present problem requires us to 
associate more information with each edge than direction.

\section{Augmentation}\label{sec:augmentation}
Given a collection $\P$ of edge-disjoint $T$-paths in a multigraph $G=(V,E)$ without selfloops,
we intend to characterize when $|\P|+1$ edge-disjoint $T$-paths exist in $G$.
We now introduce an auxiliary labeled graph $\G(\P)=((V,E\cup L), \sigma_V,\sigma_E,\sigma_L)$,
by adding certain selfloops to $G$ and assigning symbols to edges and vertices as labels.
For each $T$-path $P\in \P$, we attach two selfloops at each internal vertex of $P$.
A vertex has $2k$ selfloops if it belongs to $k$ paths in $\P$.
We denote by $L$ the set of all the attached 
selfloops.
For each $e\in E\cup L$, we denote by $\partial e$ the set of its end-vertices.
We say that an edge $e\in E\cup L$ {\em comes from} $P\in \P$ if it is used in $P$ or is a selfloop defined for $P$.

We denote by $E(\P)$ the set of edges used in $\P$.
If an edge $e\in E(\P)$ with $\partial e=\{u,v\}$ belongs to a $T$-path $P$ from $s$ to $t$,
and $s,u,e,v, t$ appear in this order along $P$,
we assign symbols $\sigma_E(e,u)\coloneqq s$ and $\sigma_E(e,v)\coloneqq t$.
No symbols are assigned to edges in $E\setminus E(\P)$.
An edge $e\in E$ is called {\em labeled} or {\em free} depending on whether $e\in E(\P)$ or not.
The two selfloops $e,\bar{e}\in L$ at a vertex $v$ coming from $P\in \P$ are assigned ordered pairs of symbols $\sigma_L(e)\coloneqq st$ and $\sigma_L(\bar{e})\coloneqq ts$, where $s$ and $t$ are terminals of $P$.
Any terminal vertex $t\in T$ is assigned $\sigma_V(t)\coloneqq t$,
and other vertices $v\in V\setminus T$ have no symbols.

A walk in $\G(\P)=((V,E\cup L), \sigma_V,\sigma_E,\sigma_L)$ is a sequence $Q=(v_0,e_1,v_1,\ldots,e_\ell,v_\ell)$ of
vertices $v_i\in V$ and edges $e_i\in E\cup L$ such that $\partial e_i=\{v_{i-1},v_i\}$ for $i=1,\ldots,\ell$.
For a walk $Q$, we associate a string $\gamma(Q)$ defined to be the sequence of symbols that appear as labels in $Q$. 
If $e_i$ is a labeled edge, then $e_i$ is assigned $\sigma_E(e_i,v_{i-1})\sigma_E(e_i,v_{i})$.
If $e_i\in L$, then $e_i$ is assigned $\sigma_L(e_i)$. If $v_i=t\in T$, then $v_i$ is assigned $\sigma_V(v_i)=t$.
Free edges and non-terminal vertices are assigned no symbols.

We call $Q$ an {\em augmenting walk} in $\G(\P)$ if $\ell>0$ and it satisfies the following
three conditions.
\begin{itemize}
\setlength{\itemsep}{0mm}
\setlength{\leftskip}{2mm}
\item[(A1)] 
The end-vertices $v_0$ and $v_\ell$ are in $T$, and intermediate vertices $v_1,\ldots,v_{\ell-1}$ are not in $T$.
\item[(A2)] The string $\gamma(Q)$ has no consecutive appearance of a symbol.
\item[(A3)] Each free edge appears at most once in $Q$.
\end{itemize}

In Figure~\ref{fig:augmentation1}, we provide an example of an augmenting walk, which suggests the significance of selfloops and double use of labeled edges.
In this case, we can augment the number of edge-disjoint $T$-paths by taking the symmetric difference 
between $E(\P)$ and $E(Q)$, where $E(Q)$ denotes the set of edges that appear in $Q$ odd number of times. 
This simple operation does not always work.
Figure~\ref{fig:augmentation2} gives an example for which the symmetric difference
gives an edge set that cannot be decomposed into edge-disjoint $T$-paths.

\begin{figure}[ht]
	\begin{center}
		\includegraphics[width=0.95\hsize]{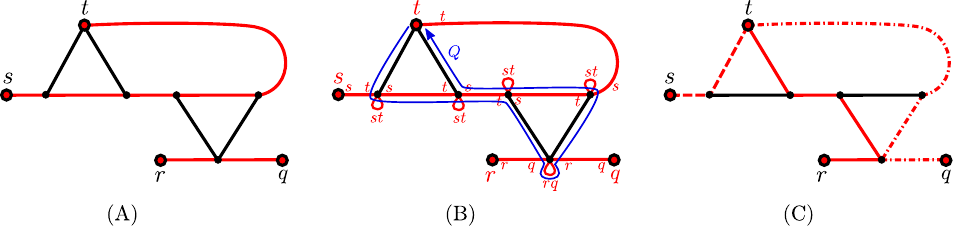}
	\end{center}
	\vspace{-4mm}
	\caption{
		\small\baselineskip=11pt
		(A) A graph $G=(V,E)$ with terminals $T=\{s,t,r,q\}$.  Red and black edges represent labeled and free edges, respectively, i.e.,
		red edges form $T$-paths $\P$.
		(B) The auxiliary labeled graph $\G(\P)$ and an augmenting walk $Q$ with $\gamma(Q)=tststrqtstst$.
		For simplicity, one of two selfloops at each vertex is omitted in the figure.
		(C) Three edge-disjoint $T$-paths in $G$.
	}
	\label{fig:augmentation1}
\end{figure}
\begin{figure}[ht]
	\begin{center}
		\includegraphics[width=0.95\hsize]{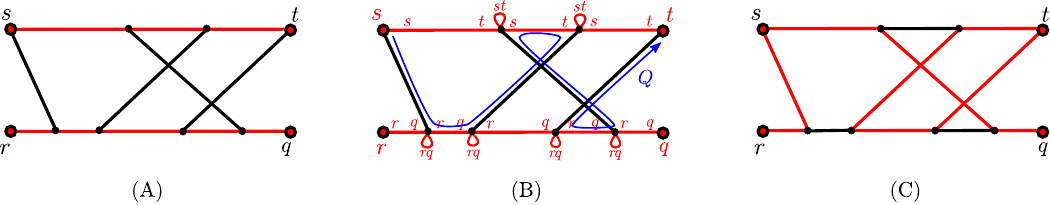}
	\end{center}
	\vspace{-4mm}
	\caption{
		\small\baselineskip=11pt
		(A) A graph $G=(V,E)$ with terminals $T=\{s,t,r,q\}$, and $T$-paths $\P$.
		(B)  The auxiliary labeled graph $\G(\P)$ and an augmenting walk $Q$, where $\gamma(Q)=srqtsqrt$.
		For simplicity, one of two selfloops at each vertex is omitted in the figure.
		(C) Red edges represent the symmetric difference of $Q$ and the $T$-paths $\P$.
		This can not be decomposed into three edge-disjoint $T$-paths
	}
	\label{fig:augmentation2}
\end{figure}

We will show that, if an augmenting walk satisfies a particular condition, called {\em shortness}, one can successfully augment edge-disjoint $T$-paths by the symmetric difference operation.

\subsection{Short Augmenting Walks}
We introduce some notations to define short augmenting walks.

For distinct vertices $u,v\in V\setminus T$ on a $T$-path $P\in \P$,
we denote by $P(u,v)$ the subpath of $P$ from $u$ to $v$.
We say that $P(u,v)$ is {\em $st$-directed} if $\gamma(P(u,v))$ is a repetition of $st$.
Also, $P(u,u)$ denotes one of the two selfloops at $u$ coming from $P$,
whose direction is specified when we use this notation.

For a walk $Q=(v_0,e_1,v_1,\ldots,e_\ell,v_\ell)$
and indices $a,b$ with $0\leq a\leq b\leq \ell$,
we denote by $Q[a,b]$ its subwalk $(v_a,e_{a+1},\ldots,e_b,v_b)$.
A subwalk $Q[a,b]$ with $a<b$ is called a {\em $P$-segment} of $Q$ if either
(i) $Q[a,b]$ is a maximal subwalk forming a subpath of $P$ or (ii) $Q[a,b]$ consists of a single selfloop coming from $P$ (hence $a+1=b$).
We call $Q[a,b]$ a {\em $\P$-segment} if it is a $P$-segment for some $P\in \P$.

Recall that Figure~\ref{fig:augmentation2} gives an example of an augmenting walk $Q$ for which the symmetric difference operation does not work.  However, we can also observe that another augmenting walk, say $Q'$, shown in Figure~\ref{fig:augmentation3} yields a successful augmentation. This $Q'$ is obtained from $Q$ by connecting two $\P$-segments coming from the same $T$-path. We now define this kind of operation formally.

\begin{figure}[t]
	\begin{center}
		\includegraphics[width=0.95\hsize]{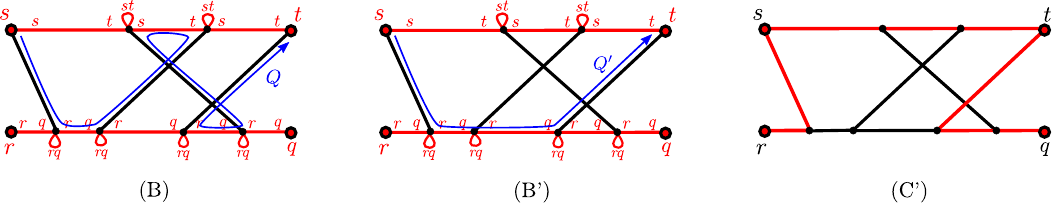}
	\end{center}
	\vspace{-4mm}
	\caption{
		\small\baselineskip=11pt
		(B)  The same figure as that in Figure~\ref{fig:augmentation2}. 
            (B') The walk $Q'$ obtained from $Q$ in (B) by bridging the two $\P$-segments coming from the bottom $T$-path. It satisfies (A2) as $\gamma(Q')=srqrqt$.
		(C) Red edges represent the symmetric difference of $Q'$ and the $T$-paths $\P$.
		This can be decomposed into three $T$-paths.
	}
	\label{fig:augmentation3}
\end{figure}

For an augmenting walk $Q=(v_0,e_1,v_1,\ldots,e_\ell,v_\ell)$, consider distinct $\P$-segments $S=Q[a,b]$ and $S'=Q[c,d]$ that come from the same $T$-path, say $P\in \P$, and appear in $Q$ in this order. 
The {\em bridging operation} applied to the pair $(S, S')$ means defining a new walk by $Q'\coloneqq Q[0,a]+P(v_a, v_d)+Q[d,\ell]$.  Note that the subpaths $S$ and $S'$ of $P$ may or may not be included in $P(v_a, v_d)$ and that the resultant walk $Q'$ is not necessarily an augmenting walk as it may violate (A2). When $Q'$ satisfies (A2), we call the pair $(S, S')$ a {\em simple shortcut}.

By requiring the nonexistece of simple shortcuts, we can exclude the augmenting walk given in Figure~\ref{fig:augmentation2}. However, it turns out that prohibiting simple shortcuts is insufficient to guarantee a successful augmentation. In Figure~\ref{fig:augmentation4}, the depicted augmenting walk admits no simple shortcuts, but the symmetric difference operation cannot augment the number of $T$-paths. 
However, if we apply the bridging operation to multiple pairs of $\P$-segments simultaneously, we obtain an augmenting walk for which symmetric difference operation works. 

\begin{figure}[h]
\begin{center}
\vspace{5mm}
\includegraphics[width=0.98\hsize]{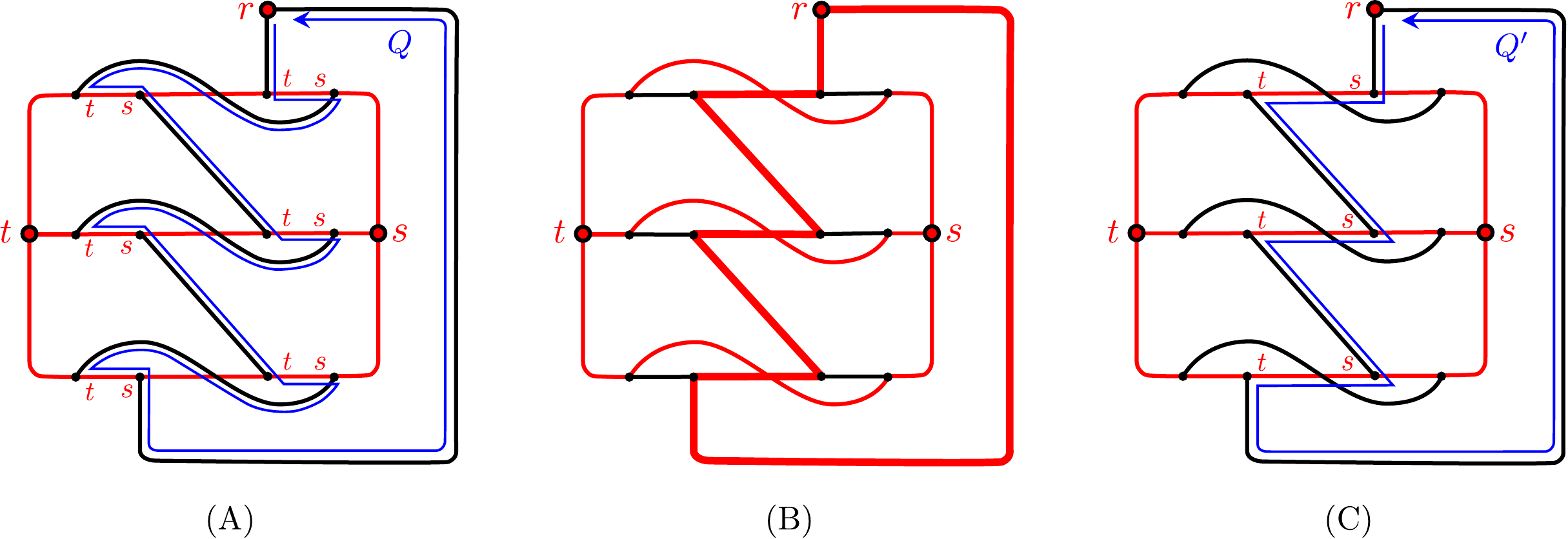}
\end{center}
\vspace{-2mm}
\caption{
\small\baselineskip=11pt
(A) The auxiliary labeled graph $\G(\P)$ and an augmenting walk $Q$, where selfloops in $\G(\P)$ are omitted in the figure for simplicity. Bridging any pair of $\P$-segments coming from the same $T$-path causes a consecutive appearance of a symbol.
(B) Red edges represent the symmetric difference of $Q$ and $\P$.
It consists of a $T$-cycle and three $T$-paths.
(C) By bridging three pairs of $\P$-segments simultaneously, we obtain an augmenting walk $Q'$.
The symmetric difference of $Q'$ and $\P$ is decomposed into four $T$-paths
}
	\label{fig:augmentation4}
\end{figure}

We then introduce a more general notion of shortcuts.
A {\em shortcut} is a collection $\{(S_i, S'_i)\}_{i=1}^k$ of pairs of $\P$-segments with $k\geq 1$ satisfying the following conditions.

\begin{itemize}
\item[(S1)] For each $i=1,2,\dots, k$, the $\P$-segments $S_i$ and $S'_i$ are coming from the same $T$-path.
\item[(S2)] The $\P$-segments $S_1, S'_1, S_2, S'_2, \dots, S_k, S'_k$ are all distinct and appear in $Q$ in this order.
\item[(S3)] The walk obtained by bridging all the pairs $\{(S_i, S'_i)\}_{i=1}^k$ satisfies (A2).
\end{itemize}
Note that a simple shortcut defined before is a shortcut with $k=1$.

In addition to the nonexistence of shortcuts, we require one more condition to define shortness, which is adopted just for the simplicity of the analysis in the subsequent parts. We say that a walk $Q$ {\em transfers} at a vertex $v$ if $v$ appears in $Q$ as the first vertex or the last vertex of some $\P$-segment. It is possible that $v$ is the last vertex of some $\P$-segment and the first vertex of another $\P$-segment (when these two $\P$-segments are consecutive), and this is also regarded as one transfer. We call an augmenting walk {\em short} if it admits no shortcut and it transfers at each vertex at most twice.
This condition is not restrictive in the sense of existence as shown in the following proposition.

\begin{proposition}
If $\G(\P)$ admits an augmenting walk, then it admits a short augmenting walk.
\end{proposition}
\begin{proof}
Suppose that  $\G(\P)$ admits an augmenting walk and let $Q=(v_0,e_1,v_1,\ldots,e_\ell,v_\ell)$ be the one that minimizes the number of $\P$-segments and, subject to that, minimizes the length $\ell$. We show that $Q$ is short.
Suppose to the contrary that $Q$ has a shortcut $\S$. Then, by bridging all the pairs in $\S$, we obtain an augmenting walk with smaller number of $\P$-segments, which contradicts the choice of $Q$.
Next, suppose to the contrary that $Q$ transfers more than twice at some vertex.
Then, $v_a=v_b=v_c$ for some indices $a<b<c$.  
Consider three walks $Q[0,a]+Q[b,\ell]$, $Q[0,a]+Q[c,\ell]$, and $Q[0,b]+Q[c,\ell]$. By the minimality of $\ell$, all of them violate (A2). From the first two, we obtain that the last symbol of $\gamma(Q[0,a])$ coincides with the first symbols of $\gamma(Q[b,\ell])$ and $\gamma(Q[c,\ell])$. From the third, we obtain that the last symbol of $\gamma(Q[0,b])$ coincides with the first symbol of $\gamma(Q[c,\ell])$. These together imply that the last symbol of $\gamma(Q[0,b])$  coincides with the first symbol of $\gamma(Q[b,\ell])$, which contradicts the fact that $Q=Q[0,b]+Q[b,\ell]$ satisfies (A2).
\end{proof}

The following is a basic property of a short augmenting walk.
\begin{lemma}
\label{lem:property_of_shortness}
For a short augmenting walk $Q$ and a $T$-path $P\in\P$, no two $P$-segments in the same direction share a vertex. 
Consequently, $Q$ uses each labeled edge at most once in each direction.
\end{lemma}
\begin{proof}
Let $Q[a,b]$ and $Q[c,d]$ with $a<c$ be $P$-segments of $Q=(v_0,e_1,v_1,\ldots,e_\ell,v_\ell)$ and suppose that they are both $st$-directed. Suppose to the contrary that they share a vertex $v^*$.
Then $s,v_a,v^*,v_d,t$ appear in this order on $P$ (possibly, $v_a=v^*=v_d$), and hence $P(v_a,v_d)$ is $st$-directed (where we let $P(v_a, v_d)$ be the $st$-directed selfloop at $v$ if $v_a=v_d$). 
Then, the pair $(Q[a,b], Q[c,d])$ forms a simple shortcut of $Q$, a contradiction.
\end{proof}

\subsection{Transition Systems}
\label{sec:ts}
In the original graph $G$, a vertex in $V\setminus T$ may be incident to more than two edges in $E(\P)$,
and hence the edge set $E(\P)$ is not sufficient to represent the collection $\P$ of edge-disjoint $T$-paths.
We will describe $\P$ and the symmetric difference operation for it
using transition systems.

A walk in $G$ is called a {\em $T$-trail} if it uses each edge at most once and
its first and last vertices are in $T$ while all internal vertices are in $V\setminus T$.
In particular, a $T$-trail is called a {\em $T$-circuit} if its first and the last vertices coincide.
If a $T$-trail is not a $T$-circuit, then it contains a $T$-path in the sense of edge set inclusion.

A {\em transition} $\T(v)$ at each vertex $v\in V\setminus T$ is a set of pairs of edges incident to $v$.
The collection $\T=\{\T(v)\}_{v\in V\setminus T}$ is called a {\em transition system}. In particular, it is called 
{\em consistent} if, for every vertex $v\in V\setminus T$ , the collection $\T(v)$ consists of disjoint pairs  and any edge $e$ appearing in $\T(v)$ is either incident to some terminal in $T$ 
or appears in $\T(u)$, where $\partial e=\{v,u\}\subseteq V\setminus T$.    

For a transition system $\T=\{\T(v)\}_{v\in V\setminus T}$, let $E(\T)$ denote the set of edges that appear in $\T(v)$ at some $v\in V\setminus T$. 
If $\T$ is consistent, we see that the subgraph $H=(V,E(\T))$ is {\em inner Eulerian}, 
i.e., every inner vertex $v\in V\setminus T$ has even degree. A consistent transition system $\T$ naturally provides 
a decomposition of $H=(V,E(\T))$ into edge-disjoint $T$-trails and inner circuits, where an {\em inner circuit} 
means a closed walk that is disjoint from $T$ and uses each edge at most once.

For a transition $\T(v)$ at $v\in V\setminus T$ and a pair $\{e,f\}$ of edges incident to $v\in V\setminus T$,
we define an operation $\T(v)\blacktriangle\{e,f\}$ by 
$$\T(v)\blacktriangle\{e,f\}\coloneqq \left\{\begin{array}{ll}
\T(v)+\{e,f\} & (e,f\in E\setminus E(\T(v))) \\
\T(v)-\{e,f\} & (\{e,f\}\in\T(v)) \\
\T(v)-\{e,e'\}+\{e',f\} & (\exists e'\in E: \{e,e'\}\in \T(v), ~f\in E\setminus E(\T(v))) \\
\T(v)-\{e,e'\}-\{f,f'\}+\{e',f'\} & (\exists e',f'\in E: \{e,e'\}, \{f,f'\}\in \T(v)),
\end{array}\right.$$
where $E(\T(v))$ is the union of pairs in $\T(v)$.
We denote by $\T\blacktriangle\{e,f\}$ the collection $\{\T'(u)\}_{u\in V\setminus T}$ such that
$\T'(v)=\T(v)\blacktriangle\{e,f\}$ and $\T'(u)= \T(u)$ for every $u\in V\setminus (T\cup \{v\})$.

The family $\P$ of edge-disjoint $T$-paths determines a transition system $\T_{\P}=\{\T_{\P}(v)\}_{v\in V\setminus T}$, 
i.e., each $\T_{\P}(v)$ consists of all pairs $\{e,e'\}$ such that $e$ and $e'$ are incident to $v$ and used in some $T$-path $P\in\P$ consecutively.

For an augmenting walk $Q=(v_0,e_1,v_1,e_2,\ldots,e_\ell,v_\ell)$ in the labeled graph $\G(\P)$, 
we define a transition system $\T_Q=\{\T_{Q}(v)\}_{v\in V\setminus T}$ as follows. If $e_i\in E$, set $\hat{e}_i\coloneqq e_i$. If $e_i$ is a self-loop at $v$ 
coming from a $T$-path $P$ with $\sigma_L(e_i)=st$, let $\hat{e}_i$ denote the edge in $P$ incident 
to $v$ with $\sigma_E(\hat{e}_i,v)=s$. Then for any vertex $v\in V\setminus T$, 
let $\T_{Q}(v)$ be a collection of pairs $\{\hat{e}_{i},\hat{e}_{i+1}\}$ with $v_i=v$.
Note that pairs in $\T_{Q}(v)$ are not necessarily mutually disjoint. 
They may even coincide (and hence $\T_{Q}(v)$ is a multi-set).
Recall that $E(Q)$ denotes the set of edges that appear in $Q$ odd number of times. In particular, if $Q$ is a short augmenting walk, $E(Q)$ coincides with the set of edges that appear in $Q$ exactly once.

We then define $\T_{\P\triangle Q}=\{\T_{\P\triangle Q}(v)\}_{v\in V\setminus T}$ by 
$\T_{\P\triangle Q}\coloneqq \T_{\P}\blacktriangle\{\hat{e}_1,\hat{e}_2\}\blacktriangle\{\hat{e}_2,\hat{e}_3\}\blacktriangle\cdots\blacktriangle\{\hat{e}_{\ell-1},\hat{e}_\ell\}$.
We see that $\T_{\P\triangle Q}$ is a consistent transition system with 
$E(\T_{\P\triangle Q})=E(\P)\triangle E(Q)$, which determines a collection
of edge-disjoint $T$-trails and inner circuits. This is called a {\em switching operation} of $\P$ by $Q$.
The resulting collection of edge-disjoint $T$-trails and inner circuits are denoted by $\P\triangle Q$. (In each of Figures~\ref{fig:motivation2}, \ref{fig:augmentation1}, and \ref{fig:augmentation2}, (C) depicts $\P\triangle Q$.)

In this way,  for each vertex $v\in V\setminus T$, the collection $\T_{\P\triangle Q}(v)$ is defined from $\T_{\P}(v)$ and $\T_{Q}(v)$. 
In the following proposition, we analyze $\T_{Q}(v)$ by using $\T_{\P}(v)$ and $\T_{\P\triangle Q}(v)$, when $Q$ is assumed to be a short augmenting walk.

\begin{lemma}
\label{lem:sawtrans} 
Let $Q$ be a short augmenting walk. For any vertex $v\in V\setminus T$ and any edges $e, f$ incident to $v$, suppose that $\{e,f\}\in \T_{\P\triangle Q}(v)$ holds. Then the following {\em (\rn{1})--(\rn{4})} hold. 
\begin{itemize}
\item[{\em (\rn{1})}] If both $e$ and $f$ are free edges, then either $\{e,f\}\in \T_Q(v)$ or there exists a pair $\{g,g'\}\in \T_{\P}(v)$ such that $\{e,g\}, \{f,g'\}\in \T_Q(v)$.
\item[{\em (\rn{2})}] If $\{e,f\}\in \T_\P(v)$, then $\T_Q(v)$ contains $\{e,f\}$ twice or neither $e$ nor $f$ appears in $\T_Q(v)$. 
\item[{\em (\rn{3})}] If $e$ is free, $f$ is labeled, and $\{f,f'\}\in \T_\P(v)$, then $\T_Q(v)$ contains $\{e,f\}$ or $\{e,f'\}$ or there exists a pair $\{g,g'\}\in \T_{\P}(v)$ such that $\{e,g\}, \{\tilde{f},g'\}\in   \T_Q(v)$ for some $\tilde{f}\in\{f,f'\}$. 
\item[{\em (\rn{4})}] If $e$ and $f$ are labeled and $\{e,e'\}, \{f,f'\} \in \T_\P(v)$, then $\T_Q(v)$ contains $\{e,f\}$, $\{e,f'\}$, $\{e',f\}$, or $\{e',f'\}$ or there exists a pair $\{g,g'\}\in \T_{\P}(v)$ such that 
$\{\tilde{e},g\}, \{\tilde{f},g'\}\in \T_Q(v)$ for some $\tilde{e}\in \{e,e'\}$ and $\tilde{f}\in\{f,f'\}$.  
\end{itemize}
\end{lemma}
\begin{proof}
Before showing (\rn{1})--(\rn{4}), we prepare two observations. 
\begin{itemize}
\setlength{\leftskip}{20mm}
\item[Observation 1.~] If $g$ and $h$ are free edges with $\{g,h\}\in \T_Q(v)$, then $\{g,h\}\in \T_{\P\triangle Q}(v)$. 
\item[Observation 2.~] For any pair $\{g,g'\}\in \T_{\P}(v)$, at most two pairs in $\T_Q(v)$ contains $g$ or $g'$.
\end{itemize}
The first one easily follows from the fact that each free edge appears at most once in $Q$.
To see the second one, suppose to the contrary that there are three pairs in $\T_Q(v)$ containing $g$ or $g'$. 
Then, there are three $\P$-segments coming from the same $T$-path and containing $v$. At least two of them have the same direction, which contradicts Lemma~\ref{lem:property_of_shortness}.

\medskip
 (\rn{1}):  Since $e$ and $f$ are free edges with $\{e,f\}\in \T_{\P\triangle Q}(v)$, each of them appears in $\T_Q(v)$ exactly once. 
Suppose $\{e, f\}\not \in \T_Q(v)$.
Then, $\{e,g\}, \{f, h\}\in \T_Q(v)$ for some edges $g$ and $h$ with $g, h\not\in \{e,f\}$. 
By Observation 1, $g$ and $h$ are labeled edges. 
Let $g'$ be the edge with $\{g, g'\}\in \T_{\P}(v)$.
We intend to show $h=g'$. 
Since $Q$ transfers at $v$ at most twice, any pair in $\T_Q(v)$ other than $\{e,g\}$ and $\{f,h\}$ is either consisting of free edges or belonging to $\T_{\P}(v)$. 
Therefore, if $h\not\in\{g, g'\}$, then $e$ is paired with $g$ or $g'$ in $\T_{\P\triangle Q}(v)$, a contradiction. Thus, we have $h\in \{g, g'\}$. In particular, we have $h=g'$. 
To see this, suppose to the contrary that $h=g$. Then, $\{e,g\}, \{f, g\}\in \T_Q(v)$ and any other pair in $\T_Q(v)$ contains neither $g$ nor $g'$ by Observation 2.  This implies $\{e,g'\}, \{f, g\}\in \T_{\P\triangle Q}(v)$, a contradiction.
Thus, we obtain $h=g'$, and hence $\{e,g\}, \{f,g'\}\in \T_Q(v)$.

\medskip
(\rn{2}): Since $\{e, f\}$ belongs to $\T_{\P}(v)$ and $\T_{\P\triangle Q}(v)$, each of $e$ and $f$ appears in $Q$ twice or never. If both of them appear twice, then $\T_Q(v)$ contains $\{e,f\}$ twice by Observation 2. It then suffices to show that it never happens that one of $e$ and $f$ appears twice in $Q$ and the other never appears. Suppose, to the contrary, that $e$ appears twice while $f$ never appears. Then, there are pairs $\{e,g\}, \{e,h\}\in \T_Q(v)$ for some edges $g, h$ distinct from $f$. As $Q$ transfers at $v$ at most twice, any pair in $\T_Q(v)$ other than $\{e,g\}$ and $\{e,h\}$ either consists of free edges or belongs to $\T_{\P}(v)$. If $g$ is a free edge, then $g$ is paired with $e$ or $f$ in $\T_{\P\triangle Q}(v)$, a contradiction. Then, $g$ is labeled. 
Let $g'$ be the edge with $\{g, g'\}\in \T_{\P}(v)$. If $h\not\in \{g,g'\}$, then $f$ is paired with $g$ or $g'$ in $\T_{\P\triangle Q}(v)$, a contradiction. If $h=g$, then $\T_Q(v)$ contains $\{e,g\}$ twice and any other pair in $\T_Q(v)$ contains neither $e$, $f$, $g$ nor $g'$ by Observation 2.  This implies $\{f,g'\}, \{e,g\}\in \T_{\P\triangle Q}(v)$, a contradiction. Thus, $g'=h$ must hold and $\{e,g\}, \{e,g'\}\in \T_Q(v)$ follows. Let $P\in \P$ be the $T$-path containing $g$ and $g'$. 
By Lemma~\ref{lem:property_of_shortness}, $e$ must be used once in each direction in $Q$.
This implies that the two $P$-segments containing $g$ and $g'$ have the same direction, which contradicts Lemma~\ref{lem:property_of_shortness}.

\medskip
 (\rn{3}): Since $e$ is free and appearing in $\T_{\P\triangle Q}(v)$, it appears exactly once in $Q$. 
Suppose that $\T_Q(v)$ contains neither $\{e,f\}$ nor $\{e, f'\}$. Then, $\{e,g\}\in \T_Q(v)$ for some edge $g$ with $g\not\in \{f, f'\}$, which is a labeled edge by Observation~1. 
Let $g'$ be an edge such that $\{g, g'\}\in \T_{\P}(v)$.
Since $f$ is not paired with $f'$ in $\T_{\P\triangle Q}(v)$ while $\{f, f'\}\in \T_{\P}(v)$, some pair in $\T_Q(v)$ contains exactly one of $f$ and $f'$, denoted by $\tilde{f}$. That is,  $\{\tilde{f}, h\}\in \T_Q(v)$ for some $h\not\in \{f,f'\}$.
Since $Q$ transfers at $v$ at most twice, any pair in $\T_Q(v)$ other than $\{e,g\}$ and $\{\tilde{f},h\}$ either consists of free edges or belongs to $\T_{\P}(v)$. 
If $h\neq g'$, then $e$ is paired with $g$ or $g'$ in $\T_{\P\triangle Q}(v)$, a contradiction.
Thus, we have $\{e,g\}, \{\tilde{f}, g'\}\in\T_Q(v)$. 

\medskip
 (\rn{ 4}): Suppose that $\T_Q(v)$ contains none of $\{e,f\}$, $\{e,f'\}$, $\{e',f\}$, nor $\{e',f'\}$. 
Since $e$ (resp., $f$) is not paired with $e'$ (resp., $f'$) in $\T_{\P\triangle Q}(v)$, there exist edges $g\not\in \{e,e'\}$ and $h\not\in \{f,f'\}$ such that $\{\tilde{e}, g\}, \{\tilde{f}, h\}\in \T_Q(v)$ for some $\tilde{e}\in \{e, e'\}$ and $\tilde{f}\in \{f, f'\}$. Since $Q$ transfers at $v$ at most twice, any pair in $\T_Q(v)$ other than $\{\tilde{e},g\}$ and $\{\tilde{f},h\}$ either consists of free edges or belongs to $\T_{\P}(v)$. 
In addition, $\T_Q(v)$ contains $\{e,e'\}$ if and only if $\tilde{e}=e$ since the labeled edge $e$ appears in $Q$ twice or never. 
If $g$ is a free edge, then $g$ is paired with $e$ in $\T_{\P\triangle Q}(v)$, a contradiction. 
Thus, $g$ is a labeled edge.
Let $g'$ be the edge with $\{g, g'\}\in \T_{\P}(v)$. If $h\neq g'$, then $e$ is paired with $g$ or $g'$ in $\T_{\P\triangle Q}(v)$, a contradiction.
Thus, we have $\{\tilde{e},g\}, \{\tilde{f},g'\}\in \T_Q(v)$.
\end{proof}
In the second assertion of (i), (iii), and (iv) of Lemma~\ref{lem:sawtrans}, if $\{g,g'\}\in \T_{\P}(v)$ comes from a $T$-path $P\in \P$, we say that $v$ is a {\em junction of $Q$ and $P$ pairing $e$ and $f$}. 

\subsection{Validity of Augmentation}
As mentioned in Section~\ref{sec:ts}, for an augmenting walk $Q$, the collection
$\T_{\P\triangle Q}$ is a transition system
of the subgraph $H^*\coloneqq (V,E(\P)\triangle E(Q))$, which determines a collection $\P\triangle Q$
of edge-disjoint $T$-trails and inner circuits.
Since the first and last edges of $Q$ are distinct free edges by (A1)--(A3),
the sum of the degrees of terminals in $H^*$
is $2(|\P|+1)$. Thus, $\P\triangle Q$ contains $|\P|+1$ edge-disjoint $T$-trails.

The following theorem states that, for a short augmenting walk $Q$,
none of $T$-trails in $\P\triangle Q$ are $T$-circuits. 
This immediately implies that $H^*$ includes $|\P|+1$ edge-disjoint $T$-paths.

\begin{theorem}\label{thm:SAW1}
For a family $\P$ of edge-disjoint $T$-paths, if $Q$ is a short augmenting walk,
then $\P\triangle Q$ contains $|\P|+1$ edge-disjoint $T$-trails none of which are $T$-circuits.
\end{theorem}
\begin{proof}
We intend to show that, if $\P\triangle Q$ contains a $T$-circuit, then $Q$ has a shortcut, which contradicts the shortness of $Q$.

To this end, we first observe some properties of $T$-trails in  $\P\triangle Q$.
Let $C$ be an arbitrary $T$-trail in $\P\triangle Q$.
We then define a collection $\C$ of edge-disjoint subwalks of $C$ as follows. Any maximal subwalk of $C$ forming a subpath of some $T$-path in $\P$ is a member of $\C$. A single inner vertex $v\in V\setminus T$ on $C$ is also a member of $\C$ if it is a junction of $Q$ and some $P\in \P$ pairing the two edges on $C$ incident to $v$. If $C$ starts (resp., ends) with a free edge, then the first (resp., last) terminal of $C$ is regarded as a member of $C$. These are all that $\C$ contains. 
We attach indices to the members of $\C$ so that $\C=\{C_0, C_1,\dots, C_\ell\}$ and they appear in this order along $C$.
See Figure~\ref{fig:augment} for an example.
If $\ell=0$, then $C$ coincides with some $T$-path in $\P$ and is not a $T$-circuit obviously.
We then assume $\ell\geq 1$.

For $i=0,1,\ldots,\ell$, let $x_i$ and $y_i$ denote the end-vertices of $C_i$ that appear in this order
along $C$. We have $x_i= y_i$ if $C_i$ is a single vertex.  
The subwalk of $C$ from $y_{i-1}$ to $x_i$,
denoted by $D_i$, consists of free edges for each $i=1,\ldots,\ell$.
Note that $y_{i-1}$ and $x_i$ may be identical. See Figure~\ref{fig:augment} for an example.
We say that $Q$ {\em traverses} $C_i$ (resp., $D_i$) if $C_i$ or its reverse (resp., $D_i$ or its reverse) appears in $Q$ as a subwalk. 
\begin{figure}[htbp]
    \vspace{4mm}
	\begin{center}
		\includegraphics[width=0.98\hsize]{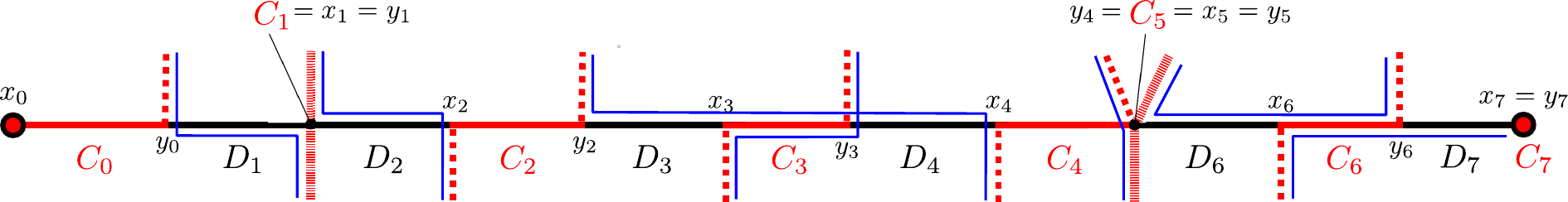}
	\end{center}
	\vspace{-0mm}
	\caption{
		\small\baselineskip=11pt
		The horizontal line represents a $T$-trail $C$ in $\P\triangle Q$, where red and black edges are labeled and free, respectively. $C_1$ and $C_5$ are junctions. Blue lines represent subwalks of a short augmenting walk $Q$. 
	}
	\label{fig:augment}
\end{figure}

\begin{claim}\label{clm:Q2}
For each $i\in\{1,2,\ldots,\ell-1\}$ with $x_i\neq y_i$, either $Q$ traverses $C_i$ twice (once in each direction) or $Q$ never uses an edge on $C_i$.
In addition, $Q$ never uses an edge on $C_0$ or $C_\ell$. 
\end{claim}
\begin{proof}
By the definition of $E(\P)\triangle E(Q)$,
any edge on each $C_i$ is used never or twice in $Q$. 
Then the first statement follows from Lemma~\ref{lem:sawtrans} (\rn{2}). 
We see that labeled edges incident to the terminal $x_0$ or $y_\ell$ cannot appear in $Q$ by (A2),
and hence the second statement follows.
\end{proof}

For $C_i\in \C$ that is not a single vertex,
let $P_i$ be the $T$-path from which $C_i$ comes.
If $C_i\in \C$ is a single vertex in $V\setminus T$, let $P_i$ be the $T$-path of which the vertex $C_i$ is a junction. 
For an end-vertex $v\in \{x_i, y_i\}$ of each $C_i$, we say that $Q$ {\em enters} (resp., {\em leaves})
$P_i$ at $v$ if there are two edges $e, f$ incident to $v$ such that $e$ and $f$ appear in $Q$ consecutively in this order and $f$ (resp., $e$) comes from $P_i$ (possibly, a selfloop) while $e$ (resp., $f$) does not. 
For convenience, we also say that $Q$ leaves (resp., enters) $P_0$ at $y_0$ if $C_0$ consists of a single terminal vertex and
$D_1$ is the first (resp., the last) part of $Q$.
Similarly, we say that $Q$ leaves (resp., enters) $P_{\ell}$ at $x_{\ell}$
if $C_\ell$ is a terminal and $D_{\ell}$ is the first (resp., the last) part of $Q$.

\begin{claim}
\label{clm:Q}
For each $i\in\{1,2,\ldots,\ell\}$, the following {\em (a)} or {\em (b)} occurs.
\begin{itemize}
\item[{\em (a)}] The walk $Q$ leaves $P_{i-1}$ at $y_{i-1}$, traverses $D_{i}$ from $y_{i-1}$ to $x_{i}$,
and enters $P_{i}$ at $x_{i}$ consecutively.
\item[{\em (b)}] The walk $Q$ leaves $P_{i}$ at $x_{i}$, traverses $D_i$ from $x_{i}$ to $y_{i-1}$,
and enters $P_{i-1}$ at $y_{i-1}$ consecutively.
\end{itemize}
\end{claim}
\begin{proof}
If $y_{i-1}\neq x_{i}$, it follows from Lemma~\ref{lem:sawtrans} (\rn{3}) that
$\T_{Q}(x_{i})$ contains a pair which consists of the free edge on $D_i$ incident to $x_{i}$ and some edge coming from $P_{i}$.
Then, $Q$ enters or leaves $P_{i}$ at $x_{i}$, and if it leaves (resp., enters), the edge immediately after leaving (resp., before entering) is the free edge on $D_i$ incident to $x_{i}$.
Similarly, we see that $Q$ enters or leaves $P_{i-1}$ at $y_{i-1}$, and
if it enters (resp., leaves), then the edge immediately before entering
(resp., after leaving) is the free edge on $D_i$ incident to $y_{i-1}$.
By Lemma~\ref{lem:sawtrans} (\rn{1}) and the definition of $\C$, consecutive edges in $D_i$ must appear in $Q$ consecutively.
Hence, either (a) or (b) occurs.

If $y_{i-1}=x_{i}$, then $D_i$ consists of the single vertex $x_i$. 
It follows from Lemma~\ref{lem:sawtrans} (\rn{3})-(\rn{4}) that $\T_{Q}(x_i)$ contains a pair
consisting of one edge from $P_{i-1}$ and one edge from $P_{i}$.
Then $Q$ leaves one of $\{P_{i-1}, P_{i}\}$ at $y_{i-1}=x_{i}$ and immediately enters the other, which means that (a) or (b) occurs.
\end{proof}
By Claim~\ref{clm:Q}, the event (a) or (b) described there occurs for each $i=1,2,\dots,\ell$.
If both (a) and (b) occur for $i$ (which happens only if $C_i$ is a single vertex), then we take one of them arbitrarily.
If (a) occurs for $i$, let $S'_{i-1}$ be the $P_{i-1}$-segment just before $Q$ leaves $P_{i-1}$ at $y_{i-1}$ and let $S_{i}$ be the $P_i$-segment just after $Q$ enters  $P_{i}$ at $x_i$. That is, the event (a) implies that $S'_{i-1}, D_i, S_i$ appear in $Q$ in this order consecutively.
Similarly, if (b) occurs for $i$, let $S'_{i-1}$ and $S_i$ be the $P_{i-1}$-segment and $P_i$-segment, respectively, such that $S_i, \overline{D_i}, S'_{i-1}$ appear in $Q$ in this order consecutively, where $\overline{D_i}$ is the reverse walk of $D_i$.

In addition, if (a) (resp., (b)) occurs for $i$,  let $s_{i-1}$ be the last (resp., first) symbol in $S'_{i-1}$ and $t_i$ be the first (resp., last) symbol in $S_{i}$. We then have $s_{i-1}\neq t_{i}$ since $D_i$ contains no symbols and $Q$ satisfies (A2).
Observe that $t_0$ or $s_\ell$ may be undefined (if $C_0$ or $C_\ell$ is a single vertex) while $s_0$ and $t_\ell$ are always defined.

We may have $S_i=S'_i$, which implies that either (1) $S'_{i-1}, D_i,S_i(=S'_i), D_{i+1}, S_{i+1}$ appear in $Q$ consecutively in this order or
(2) so do $S_{i+1}, \overline{D_{i+1}},S'_i(=S_i), \overline{D_i}, S'_{i-1}$ \,(a situation like the subwalk in Figure~\ref{fig:augment} containing $D_3$ and $D_4$).
In these cases, we redefine $S_i$ or $S'_i$ as follows.
In case (1), $S_i$ coincides with $C_i=P_i(x_i, y_i)$, which is $t_i s_i$-directed. By Claim~\ref{clm:Q2}, then $Q$ has an $s_i t_i$-directed $P_i$-segment $S_i^*$ including $\overline{C}_i\coloneqq P_i(y_i, x_i)$. We redefine $S_i$ (resp.,  $S'_i$) to be $S^*_i$ if $S^*_i$ appears in $Q$ before (resp.,  after) $S_i=S'_i$.
We say that $S_i$  (resp., $S'_i$) is  {\em adjusted} if it is updated to $S^*_i$.  After the adjustment, $S_i$ and $S'_i$ are distinct $P_i$-segments appearing in this order and the adjusted one is $s_i t_i$-directed.
In case (2), $S_i$ coincides with $\overline{C}_i$, and  we adjust $S_i$ or $S'_i$, i.e., we redefine one of them to be the $P_i$-segment including $C_i$ so that  $S'_i$ and $S_i$ are distinct $P_i$-segments appearing in this order. 

\medskip
We are now ready to show that $C$ is not a $T$-circuit.
Suppose, to the contrary, that $C$ is a $T$-circuit, i.e., the two end-vertices of $C$ are both $r$ for some terminal $r\in T$.
We now show that there exists a pair $(a,b)$ of indices with $0<a\leq b<\ell$ such that 
$\{(S_i, S'_i)\}_{i=a, a+1,\dots,b}$ or $\{(S'_i, S_i)\}_{i=b, b-1, \dots, a}$ forms a shortcut, which contradicts the shortness of $Q$.

Let $I, J\subseteq \{1,\ldots, \ell-1\}$ be the subsets of indices defined as follows.
The set $I$ is the union of $I_1, I_2, I_3\subseteq \{1,\ldots, \ell-1\}$, where
\begin{eqnarray*}
I_1 & = & \set{i\mid s_i\neq s_{i-1}}, \\
I_2 & = & \set{i\mid \text{$S_i$ and $S'_i$ appear in $Q$ in this order and $S_i$ is $s_i t_i$-directed}}, \\
I_3 & =&\set{i\mid \text{$S'_i$ and $S_i$ appear in $Q$ in this order and $S_i$ is $t_i s_i$-directed}}.
\end{eqnarray*}
Similarly, the set $J$ is the union of $J_1,J_2,J_3\subseteq  \{1,\ldots, \ell-1\}$, where
\begin{eqnarray*}
J_1 & = & \set{i\mid t_{i}\neq t_{i+1} },\\
J_2 & = & \set{i\mid \text{$S_i$ and $S'_i$ appear in $Q$ in this order and $S'_i$ is $s_i t_i$-directed}}, \\
J_3 & = & \set{i\mid \text{$S'_i$ and $S_i$ appear in $Q$ in this order and $S'_i$ is $t_i s_i$-directed}}.
\end{eqnarray*}
We let $(a, b)$ be a pair of indices that minimizes $b-a$ subject to $a\in I$, $b\in J$, and $0<a\leq b<\ell$. The existence of such a pair is guaranteed by the following claim.
\begin{claim}\label{clm:cd}
There exists a pair $(c,d)$ of indices such that $c\in I_1$, $d\in J_1$, and $0< c\leq d<\ell$.
\end{claim}
\begin{proof}
We have $\ell\geq 1$ and the condition (A2) of $Q$ implies $s_{\ell-1}\neq t_{\ell}$.
Since the two end-vertices of $C$ are both $r\in T$, the last statement of Claim~\ref{clm:Q2} implies $s_0=r$ and $t_\ell=r$.
Then $s_0\neq s_{\ell-1}$, and hence $I_1$ is nonempty. Let $c$ be the minimum index in $I_1$.  Then $r=s_{c-1}\neq t_{c}$.
Let $d$ be the largest index in $J_1$ with $c\leq d<\ell$.
Such an index $d$ exists because otherwise $r\neq t_c=t_{c+1}=\cdots=t_{\ell}=r$, a contradiction. The pair $(c,d)$ satisfies the required conditions.
\end{proof}

By Claim~\ref{clm:cd}, there exists a pair $(a,b)$ of indices with $a\in I$, $b\in J$, and $0<a\leq b<\ell$.
Among all such pairs, let $(a,b)$ be the one that attains the minimum value of $b-a$. 
Set $s\coloneqq s_a$ and $t\coloneqq t_b$.
\begin{claim}\label{clm:between}
For each index $i$ with $a\leq i\leq b$, we have $s_i=s$ and $t_i=t$.
In particular, when $x_i\neq y_i$, $C_i=P_i(x_i,y_i)$ is $st$-directed if $C_i$ is never traversed by $Q$ and $ts$-directed if $C_i$ is traversed twice.
\end{claim}
\begin{proof}
The first statement is obvious if $a=b$.
Suppose $a<b$.
By the choice of $a$ and $b$, we have $a\not\in J_1$, $b\not\in I_1$, and
$i\not\in J_1\cup I_1$ for any $i$ with $a<i<b$.
Thus, $s_i=s$ and $t_i=t$ if $a\leq i\leq b$.

To show the second statement, assume $x_i\neq y_i$. Note that Lemma~\ref{lem:property_of_shortness} implies that any vertex is contained in at most two $\P$-segments coming from the same $T$-path. Note also that $S_i$ and $S'_i$ contain vertices $x_i$ and $y_i$, respectively.
If $C_i=(x_i, y_i)$ is traversed by $Q$ twice, then each of $S_i$ and $S'_i$ includes the first or the second traversal of $C_i$ since otherwise $x_i$ or $y_i$ is contained in more than two $\P$-segments coming from the same $T$-path.
Thus, the definitions of $s_i$ and $t_i$ imply that $C_i$ is $t_i s_i$-directed, and hence it is $ts$-directed.
If $C_i=(x_i, y_i)$ is never traversed by $Q$, then $S_i$ and $S'_i$ are edge-disjoint from $C_i$.
Then, the definitions of $s_i$ and $t_i$ imply that $C_i$ is $s t$-directed.
\end{proof}


\begin{claim}\label{clm:direction}
The following {\em (\rn{1})} or {\em (\rn{2})} holds.
\begin{itemize}
\item[{\em (\rn{1})}] The $\P$-segments $S_a, S'_{a}, S_{a+1}, S'_{a+1}, \dots, S_{b}, S'_{b}$ appear in this order in $Q$. For each $i$ with $a<i\leq b$, we have that $S'_{i-1}, D_{i},  S_{i}$ appear in $Q$ in this order consecutively and both $S'_{i-1}$ and $S_i$ are $ts$-directed.
\item[{\em (\rn{2})}] The $\P$-segments $S'_b, S_{b}, S'_{b-1}, S_{b-1}, \dots, S'_{a}, S_{a}$ appear in this order in $Q$. For each $i$ with \mbox{$a<i\leq b$}, we have that $S_{i}, \overline{D_i}, S'_{i-1}$ appear in $Q$ in this order consecutively and both $S'_{i-1}$ and $S_i$ are $st$-directed.
\end{itemize}
\end{claim}
\begin{proof}
The statement is trivial if $a=b$. We now suppose $a<b$. Then, $a\not\in J$ and $b\not\in I$ hold by the minimality of $b-a$.
Since $a\not\in J_2\cup J_3$, either (1) $S_a$ and $S'_a$ appear in this order and $S'_a$ is $ts$-directed or 
(2) $S'_a$ and $S_a$ appear in this order and $S'_a$ is $st$-directed.

Here we consider case (1), which implies that $S'_a$  is not adjusted while $S_a$ may be adjusted (recall the definition of the adjustment).
Then, $S'_a, D_{a+1}, S_{a+1}$ appear in this order consecutively in $Q$, and hence $S_{a+1}$ is also $ts$-directed.

For any $i$ with $a<i<b$, we have $i\not\in  I_3\cup J_2$, and hence $S_{i}$ being $ts$-directed implies that $S_{i}$ and $S'_{i}$ appear in this order in $Q$ and $S'_{i}$ is also $ts$-directed.
This means that neither $S_i$ nor $S'_i$ is adjusted since otherwise they have opposite directions.
Thus, $S'_{i-1}, D_{i}, S_{i}$ appear in this order
consecutively in $Q$.
By applying this  implication to $i=a+1, a+2, \dots, b-1$, we obtain that 
$S'_{a}, S_{a+1}, S'_{a+1}, \dots, S_{b-1}, S'_{b-1}, S_{b}$ appear in this order in $Q$ and they are all $ts$-directed.
Since $b\not\in I_3$, the segments $S_b$ and $S'_b$ appear in this order in $Q$. As $S_b$ is $ts$-directed, it is not adjusted, and hence $S'_{b-1}, D_{b}, S_{b}$ appears in this order consecutively in $Q$. Thus, we have shown that 
(i) holds in case (1).

By the same arguments, we can show that (ii) holds in case (2).
\end{proof}

Without loss of generality, we assume that (\rn{1}) in Claim~\ref{clm:direction} holds.
(In case of~(\rn{2}), we can apply the same arguments to the reverse walk of $C$.)
We now show that $\S\coloneqq \{(S_i, S'_i)\}_{i=a, a+1,\dots,b}$ forms a shortcut of $Q$.  By the definitions of $S_i$ and $S'_i$, 
the assumption implies (S1) and (S2). Then, it suffices to show (S3).

By applying the bridging operation to $(S_i, S'_i)$, the subwalk of $Q$ from $S_i$ to $S'_i$ (including them) is replaced with the subpath of $P_i$ from the first vertex of $S_i$ to the last vertex of $S'_i$. This inserted subpath of $P_i$ is $st$-directed for any $i=a, a+1,\dots,b$ as follows.
If $C_i$ is never traversed by $Q$, then $C_i=P_i(x_i, y_i)$ is $st$-directed by Claim~\ref{clm:between}. As $S_i$ and $S'_i$ are incident to $x_i$ and $y_i$, respectively, the inserted subpath coincides with $C_i$, and hence is $st$-directed. If $C_i$ is traversed twice by $Q$, then $C_i$ is $ts$-directed by Claim~\ref{clm:between} and $S_i$ and $S'_i$ traverse $C_i$ in the opposite directions each other.  Since $S_i$ or $S'_i$ is not adjusted, $S_i$ starts at $x_i$ or $S'_i$ ends at $y_i$. In each case, the subpath of $P_i$ from the first vertex of $S_i$ to the last vertex of $S'_i$ is $st$-directed (possibly a selfloop assigned with $st$).

By (\rn{1}), $S'_{i-1}, D_{i}, S_{i}$ appear in this order consecutively in $Q$ for each $i=a+1,\dots,b$. Therefore, by bridging all pairs in $\S\coloneqq \{(S_i, S'_i)\}_{i=a, a+1,\dots,b}$, the subwalk of $Q$ from $S_a$ to $S'_b$ is replaced with the walk consisting of $st$-directed $\P$-segments and free edges.
To obtain (S3), then it suffices to show that the last symbol in $Q$ before $S_a$ is not $s$ and the first symbol after $S'_b$ is not $t$.
Note that  (\rn{1}) implies $a\not\in I_3$ and $b\not\in J_3$.
If $S_a$ is $st$-directed, then clearly the last symbol before $S_a$ is not $s$. 
If $S_a$ is $ts$-directed, then $a\not\in I_2\cup I_3$, and hence $a\in I_1$.
Also, $S_a$ being $ts$-directed implies that $Q$ leaves $P_{a-1}$ and traverses $D_i$ just before it enters $S_a$, and hence  the last symbol before $S_a$ is $s_{a-1}$, where we have $s_{a-1}\neq s_a=s$ by $a\in I_1$. Thus, the last symbol in $Q$ before $S_a$ is not $s$ in any case.
We can similarly show that the first symbol in $Q$ after $S'_b$ is not $t$. 
Therefore, $\S$ satisfies (S3). 
\end{proof}

\section{Search for a Short Augmenting Walk} \label{sec:search}
In this section, we design a search algorithm to find a short augmenting walk.
Basically, our algorithm is a construction of search trees rooted at terminals, where a special operation is applied whenever it detects a structure called a blossom.
Each tree is constructed so that, if we shrink all blossoms, then any path on the tree from the root satisfies (A2). 
Since we should avoid shortcuts, it seems better to grow the trees using few long subpaths of the current $T$-paths rather than many short subpaths.
To implement this idea, we add to $\G(\P)$ new labeled edges representing all subpaths of $T$-paths in $\P$.

Formally, we  define a labeled graph $\G^*(\P)=((V,E^*\cup L), \sigma_V, \sigma_{E^*}, \sigma_L)$, which is a supergraph of $\G(\P)=((V,E\cup L), \sigma_V,\sigma_E,\sigma_L)$. The edge set $E^*$ is obtained by adding new edges to $E$ as follows. For each pair of vertices $u$ and $v$ on a $T$-path $P\in \P$ such that $P(u,v)$ contains more than one edge, 
we attach a {\em jumping edge} $e$ with $\partial e=\{u,v\}$. 
The labels associated with this jumping edge $e$ are defined by $\sigma_{E^*}(e,u)\coloneqq s$ 
and $\sigma_{E^*}(e,v)\coloneqq t$, provided that the subpath $P(u,v)$ is $st$-directed. 
Other labels in $\G^*(\P)$ remain the same as those in $\G(\P)$.

We call a walk $Q=(v_0,e_1,v_1,\ldots,e_\ell,v_\ell)$ in the labeled graph $\G^*(\P)$ {\em admissible} if it satisfies the following condition in addition to (A2) and (A3).
\begin{itemize}
	\setlength{\itemsep}{0mm}
	\setlength{\leftskip}{2mm}
	\item[(A1*)]
	The initial vertex $v_0$ is in $T$, and intermediate vertices $v_1,\ldots,v_{\ell-1}$ are not in $T$.
\end{itemize}
In contrast to (A1), this does not require $v_\ell\in T$.
We apply the definition of shortcuts in Section~\ref{sec:augmentation} to admissible walks in $\G^*(\P)$. 

Note that any admissible walk in $\G^*(\P)$ is identified with a walk in $\G(\P)$ by replacing each jumping edge with the corresponding subpath in $\P$. We can easily see that the resultant walk is admissible in $\G(\P)$. In addition, this transformation does not change the end-vertices of each $\P$-segment. 
The definition of shortness then implies the following observation.
\begin{lemma}
\label{lem:saw_2}
A short augmenting walk in $\G^*(\P)$ corresponds to a short augmenting walk in $\G(\P)$.
\end{lemma}
In the rest of this section, we present an algorithm that finds a short augmenting walk in $\G^*(\P)$. 

\subsection{Algorithm Description}\label{sec:alg}
The algorithm works on $\G^*(\P)=((V,E^*\cup L), \sigma_V, \sigma_{E^*}, \sigma_L)$ and maintains a forest $F$ in $G^*\coloneqq (V,E^*)$ and a laminar family $\B$ on $V\setminus T$.
The forest $F$, which represents the search history, consists of $|T|$ disjoint trees rooted at $T$.
We denote by $V(F)$ and $E^*(F)$ the vertex and edge sets of $F$, respectively.
For any $v\in V(F)\setminus T$, a vertex on the path from $T$ to $v$ in $F$ is called an {\em ancestor} of $v$,
and the ancestor adjacent to $v$ in $F$ is the {\em parent} of $v$.
The edge connecting $v$ and its parent is called the {\em stalk} of $v$.
A vertex $u$ is called a {\em descendant} of $v$ if $v$ is an ancestor of $u$. Note that $v$ itself is also regarded as an ancestor and a descendant of $v$.

Each member $B$ of $\B$, called a {\em blossom}, is a subset of $V(F)\setminus T$ such that the subgraph $F[B]$ of $F$ induced by $B$ forms a subtree of $F$. We denote by $V(\B)$ the union of all members of $\B$.
Each blossom $B$ is associated with an edge $e_B\in E^*\cup L$ such that $e_B\not\in E^*(F)$ and $\partial e_B\subseteq B$, i.e., $e_B$ is spanned by the subtree $F[B]$.
We call the root vertex of $F[B]$ the {\em calyx} of $B$ and denote it by $w_B$.
As will be shown in Lemma~\ref{lem:subtree}, our algorithm constructs $F$ and $\B$ so that the following property is preserved.
\begin{itemize}
	\setlength{\itemsep}{0mm}
	\setlength{\leftskip}{2mm}
	\item[(B)] For each blossom $B\in\B$, the induced subgraph $F[B]$ is a subtree of $F$ disjoint from $T$,
	and the stalk of the calyx $w_B$ is a free edge.
\end{itemize}

In the algorithm, each vertex $v\in V(F)$ on the tree rooted at $t\in T$ stores a walk $W_1(v)$ from $t$ to $v$.
In addition, if $v\in V(\B)$, then $v$ also stores another walk $W_2(v)$ from $t$ to $v$.
These walks are called {\em search walks}. In particular, $W_1(v)$ and $W_2(v)$ are respectively called the {\em primary} and {\em secondary} search walk of $v$.
As will be shown later, search walks consist of edges in $E^*(F)\cup \set{e_B | B\in \B}$ and satisfy
$\lambda(W_1(v))\neq \lambda(W_2(v))$,
where $\lambda(W)$ denotes the last symbol of $\gamma(W)$ for a walk $W$.
\medskip

Our algorithm can be seen as a generalization of the breadth-first search. 
It manages vertices to check in two queues, denoted by $\Phi_1$ and $\Phi_2$.
A vertex $v$ is put into the {\em primary queue} $\Phi_1$ when $v$ is added to $F$.
At this moment, the primary search walk $W_1(v)$ is found.
Subsequently, $v$ is put into the {\em secondary queue} $\Phi_2$ if a blossom containing $v$ is produced for the first time.
At this moment, the secondary search walk $W_2(v)$ is found, which can be seen as a backup walk as follows. 

Recall that we have to search avoiding a consecutive appearance of a symbol, i.e., preserving (A2). For a labeled edge $e$ incident to $v$, it can happen that appending $e$ to $W_1(v)$ violates (A2) due to $\lambda(W_1(v))=\sigma_{E^*}(e,v)$ while appending $e$ to some other admissible walk ending at $v$ preserves (A2). The secondary search walk $W_2(v)$ is used to continue the search through $e$ in such a case. The condition $\lambda(W_1(v))\neq \lambda(W_2(v))$ is required for this purpose.

While either of the two queues is nonempty, the algorithm takes a vertex out of them with priority on the secondary queue $\Phi_2$. In other words, the algorithm takes the first vertex out of $\Phi_2$ if it is nonempty. Otherwise, it takes the first vertex out of $\Phi_1$. 
The algorithm scans all edges incident to the vertex $v$ taken out of the queues.
Here we introduce three types of edges the algorithm should detect.
We use the symbol $\cdot$ to represent a concatenation and denote by $\overline{W}$ the reverse walk of a walk $W$.
Let $e\in E^*\cup L$ be an edge in $\G^*(\P)$ with $\partial e=\{v,u\}$.
\begin{itemize}
	\setlength{\leftskip}{-3mm}
	\item The edge $e$ is {\em frontier} if $v\in V(F)$, $u\not\in V(F)$, and $W_i(v)\cdot e\cdot u$ satisfies (A2) for some $i\in\{1,2\}$.
	\item The edge $e$ is {\em exterior} if 
	$v\in V(\B)$, $u\in V(F)\setminus V(\B)$, and $e\in E^*(F)$ is a free edge that is the stalk of $u$.
	\item The edge $e$ is {\em interior} if $u,v\in V(F)$, $e\not\in E^*(F)$, there is no blossom $B\in \B$ with $\{u,v\} \subseteq B$, and the walk $W_i(v)\cdot e\cdot \overline{W_j(u)}$ satisfies (A2) for some $i,j\in \{1,2\}$.
\end{itemize}
In the definitions of frontier and interior edges,  $i$ and $j$ should be indices for which $W_i(v)$ and $W_j(u)$ have been determined.
Note that an interior edge may be a selfloop at $v=u$ while a frontier edge and an exterior edge must belong to $E^*$ by their definitions.

Consider an interior edge $e$ and take $i$ and $j$ in the definition such that $i+j$ is minimized.
If $v$ and $u$ have no common ancestor in $V\setminus T$, then
the algorithm returns a walk $W_i(v)\cdot e\cdot \overline{W_j(u)}$,
which is a short augmenting walk as will be shown in Theorem~\ref{thm:search-main}.
Otherwise, an interior edge $e$ {\em produces} a new blossom 
as follows.

Among the common ancestors of $v$ and $u$, 
let $w$ be the furthest one from the common root subject to the constraint that 
the stalk of $w$ is a free edge.
Such a vertex $w$ does exist in $V\setminus T$ because the edges on $F$ incident to $T$ should be free edges in the algorithm.
Let $Y$ be the set of vertices on the $w$--$v$ and $w$--$u$ paths on $F$.
A new blossom $B$ produced by $e$ is the union of $Y$ and the members of $\B$ intersecting with $Y$.
\begin{figure}[ht]
	\begin{center}
		\includegraphics[width=0.27\hsize]{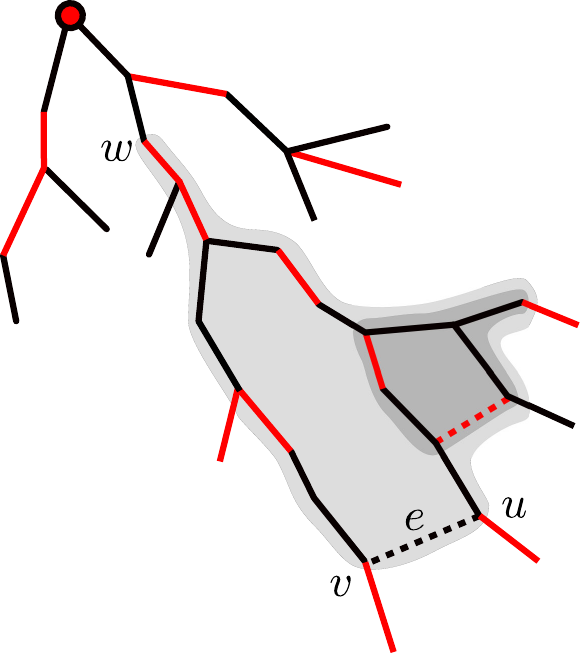}
	\end{center}
	\vspace{-2mm}
	\caption{
		\small\baselineskip=11pt
		The light gray part represents a blossom $B$ produced by an interior edge $e$ with $\partial e=\{v,u\}$.
		The dark gray part is some other blossom produced before $B$ and is included by $B$.
	}
	\label{fig:blossom}
\end{figure}
The algorithm associates the edge $e$ to $B$, i.e., sets $e_B\coloneqq e$. 
See Figure~\ref{fig:blossom}.
When a blossom $B$ is produced, the algorithm puts all vertices in $B\setminus V(\B)$ into the secondary queue $\Phi_2$ in an {\em ascending order} with respect to $F$, i.e., an ancestor $x\in B\setminus V(\B)$ of $y\in B\setminus V(\B)$ is put into $\Phi_2$ after $y$. 

When our search algorithm takes a vertex $v$ out of the primary or secondary queue,
the following procedure $\scan(v, e)$ is called for edges $e\in E^*\cup L$ incident to $v$.
It describes how to update the forest $F$ and the blossom family $\B$ when $e$ is detected as a frontier, exterior, or interior edge.
The notation $W_i(y,v)$ in Case~(\rn{3}) means a subwalk of $W_i(v)$ from $y$ to the end.
We will show in the proof of Lemma~\ref{lem:shrink} that, whenever this notation is used, $y$ appears exactly once in $W_i(v)$,
which ensures that $W_i(y,v)$ is well-defined.

\makeatletter
\renewcommand{\ALG@name}{Procedure}
\begin{algorithm}[h]
	\caption{~$\scan(v,e)$}
	\begin{enumerate}
		\setlength{\itemsep}{0mm}
		\setlength{\leftskip}{0mm}
		\item[] Let $u$ be the vertex such that $\partial e=\{v,u\}$ (possibly, $v=u$). 
		Check whether $e$ is either a  frontier, exterior, or interior edge and apply the following operations accordingly.
			\begin{description}
				\setlength{\leftskip}{0mm}
				\setlength{\itemsep}{2mm}
				\item[Case (\rn{1}):  A frontier edge.] Let $i\in \{1,2\}$ be the smallest index with $W_i(v)\cdot e\cdot u$ satisfying (A2). 
				Set $W_1(u)\coloneqq W_i(v)\cdot e\cdot u$, add $e$ and $u$ to $F$, and put $u$ into the primary queue~$\Phi_1$.
				\item[Case (\rn{2}): An exterior edge.] 
				Set $W_2(u)\coloneqq W_2(v)\cdot e\cdot u$, add $u$ to all $B\in \B$ with $v\in B$, and
				put $u$ into the secondary queue $\Phi_2$.
				\item[Case (\rn{3}): An interior edge.]
				Let $i, j\in \{1,2\}$ be indices that minimize $i+j$ subject to  \mbox{$W_i(v)\cdot e\cdot \overline{W_j(u)}$} satisfying (A2).
				\begin{itemize}
				\setlength{\leftskip}{0mm}
				\item If $v$ and $u$ do not share a common ancestor in $V\setminus T$,
				then return $W_i(v)\cdot e\cdot \overline{W_j(u)}$ and halt.
				\item Otherwise, let $B$ be the new blossom produced by $e$ and set $e_B\coloneqq e$.
				For each vertex $y\in B\setminus V(\B)$, set
                $$W_2(y)\coloneqq \left\{\begin{array}{ll}
                W_i(v)\cdot e\cdot \overline{W_j(y,u)} & (\mbox{if $y$ is an ancestor of $u$}), \\
                W_j(u)\cdot e\cdot \overline{W_i(y,v)} & (\mbox{otherwise}).
                \end{array}\right.$$
                Put ancestors of $u$ in $B\setminus V(\B)$ into $\Phi_2$ in the ascending order of $F$ and then put ancestors of $v$ in  $B\setminus V(\B)$ into $\Phi_2$ in the ascending order of $F$. Add $B$ to $\B$.
				\end{itemize}
				\end{description}
	\end{enumerate}
\end{algorithm}

Using this procedure, our search algorithm $\search$ works as described in Algorithm~\ref{algo:search}.
Here we provide a simple version that requires $O(|E^*|)=O(|E|^2)$ calls of $\scan$. In Section~\ref{sec:linear-scan}, we will show that one can reduce the number of scans to $O(|E|)$ by utilizing pointers that move on the $T$-paths in $\P$.

\makeatletter
\renewcommand{\ALG@name}{Algorithm}
\begin{algorithm}[h]
	\caption{~$\search$}
	\begin{description}
		\setlength{\itemsep}{-1mm}
		\item[Input:] A graph $G=(V, E)$ and edge-disjoint $T$-paths $\P$.
		\item[Output:] A short augmenting walk in $\G^*(\P)$ or a message ``no augmenting walk.''
	\end{description}
	\begin{enumerate}
		\setlength{\itemsep}{0mm}
		\setlength{\leftskip}{-1mm}
		\item 	Set $F\gets(T,\emptyset)$ and $\B\gets\emptyset$.
		For each terminal vertex $t\in T$, set $W_1(t)\coloneqq  t$ and put $t$ into the primary queue $\Phi_1$. 
		\item While the queue $\Phi_1$ or $\Phi_2$ is nonempty, do the following.
		\begin{enumerate}
		\item[(2-1)] If the secondary queue $\Phi_2$ is nonempty, then take the first element out of $\Phi_2$. Otherwise, take the first element out of $\Phi_1$. Let $v$ be the vertex thus dequeued. 
		\item[(2-2)] For each edge $e\in E^*\cup L$ incident to $v$, call $\scan(v,e)$.
		\end{enumerate}
		\item  Return ``no augmenting walk.'' 
	\end{enumerate}
	\label{algo:search}
\end{algorithm}

At Step 1, the algorithm starts with $F=(T, \emptyset)$.
Note that $F$ is updated only when Case (i) of $\scan(v, e)$ is applied. 
By the definition of a frontier edge, we see that  $F$ is indeed a forest rooted at $T$ at any moment of the algorithm. 
Note that any labeled edge $e\in E^*$ incident to a terminal $t\in T$ satisfies $\sigma_{E^*}(e,t)=t$, and hence it cannot be a frontier edge. 
Therefore, at any moment of the algorithm, all edges in $E^*(F)$ incident to $T$ are free edges.

The algorithm also starts with $\B=\emptyset$, which is a laminar family. 
When $\B$ is updated in Case~(\rn{2}) of $\scan(v,e)$, each blossom $B\in\B$ that contains $v$ is replaced by $B\cup\{u\}$. 
Since $u$ was not contained in any blossom before this update, the family $\B$ remains to be laminar.
When the algorithm produces a new blossom $B$ in Case~(\rn{3}), any other blossoms that intersect
with $B$ are included in $B$. Thus the algorithm maintains $\B$ as a laminar family.
The family $\B$ also satisfies the property (B) as follows.
\begin{lemma}\label{lem:subtree}
The family $\B$ of blossoms satisfies {\em (B)} throughout the algorithm.
\end{lemma}
\begin{proof}
The family $\B$ is empty at the beginning of Step 2 and updated in Cases (\rn{2}) and (\rn{3}) of $\scan(v,e)$.
The update in Case (\rn{2}) just adds a new vertex $u$ to all the blossoms that contain $v$, which is the parent of $u$. 
Then, (B) is preserved.
We next consider Case (\rn{3}).
Let $\B$ denote the laminar family before $B$ is added and define $w$ as in the definition of a new blossom mentioned above.
Because all blossoms in $\B$ form subtrees of $F$ by the inductive assumption,
the definition of $B$ implies that the induced subgraph $F[B]$ is connected, and therefore it forms a subtree of $F$.
If $w\not\in V(\B)$, then the calyx of $B$ is $w$, and the stalk of $w$ is free by definition.
Otherwise, let $B'\in \B$ be the maximal blossom with $w\in B'$.
Then, the calyx of $B$ coincides with the calyx $w_{B'}$ of $B'$,
where the stalk of $w_{B'}$ is free by the inductive assumption. Thus, (B) is preserved.
\end{proof}

\subsection{Admissibility}\label{sec:admissible}
Obviously from the description of $\search$, each vertex $x\in V(F)$ is associated with the primary search walk $W_1(x)$.
In addition, if $x\in V(\B)$, then $x$ is also associated with the secondary search walk $W_2(x)$.
By the description, any search walk determined in the algorithm consists of edges in $E^*(F)\cup \set{e_B|B\in \B}$ and satisfies {\rm (A1)}.
In this section, we show that any search walk is admissible and that $\lambda(W_1(x))\neq \lambda(W_2(x))$ holds for each $x\in V(\B)$, i.e., 
the primary and secondary search walks of $x$ have different last symbols.

For a blossom $B\in\B$, {\em shrinking} $B$ means contracting all edges $e\in E^*(F)$ with $\partial e\subseteq B$. 
Let $F/\B$ denote the graph obtained from $F$ by shrinking every member of $\B$. Since the induced subgraph $F[B]$ on any blossom $B\in\B$ forms a subtree disjoint from $T$ by (B), the resulting $F/\B$ is again a forest rooted at $T$. 
For a vertex $x\in V(F)$, its {\em projection} in $F/\B$ is the vertex $x$ itself if $x\not\in V(\B)$ and otherwise the shrunk vertex corresponding to the maximal blossom containing $x$.
For any search walk $W_k(x)$, we denote by $W_k(x)/\B$ the walk on $F/\B$ obtained as the projection of $W_k(x)$. 
We now show that any search walk satisfies the following property.
\begin{itemize}
	\item[($\star$)] 
	For any search walk $W_k(x)$, 
	the walk $W_k(x)/\B$ is the unique path on $F/\B$ from $T$ to the projection of $x$.
\end{itemize}
\begin{lemma}\label{lem:shrink}
The property {\rm($\star$)} is preserved at any moment of the algorithm.
\end{lemma}
\begin{proof}
At the beginning of Step 2, any search walk consists of a single vertex, and hence ($\star$) holds.
We show by induction that updates in Cases (\rn{1})--(\rn{3}) of $\scan(v,e)$ preserves ($\star$).

In Cases (i) and (ii), the new search walk is in the form of $W_i(v)\cdot e\cdot u$, where $e$ is the stalk of $u$ by the definitions of frontier and exterior edges. Since $W_i(v)$ satisfies the property in ($\star$) by the inductive assumption, the new walk also satisfies the property.

In Case (\rn{3}), a new search walk $W_2(y)$ is determined for $y\in B\setminus V(\B)$.
Before proving the property in ($\star$), we show that $W_2(y)$ is indeed well-defined. 
First, suppose that $y$ is an ancestor of $u$. 
Since $W_j(u)$ satisfies the property in ($\star$) by the inductive assumption and $y\not\in V(\B)$ holds, 
the ancestor $y$ of $u$ appears exactly once in $W_j(u)$.
Thus, $W_j(y,u)$ is well-defined, and so is $W_2(y)=W_i(v)\cdot e\cdot \overline{W_j(y,u)}$.
Similarly, in case $y$ is not an ancestor of $u$ but that of $v$, we see that $v$ appears exactly once in $W_i(v)$, and hence $W_2(y)=W_j(u)\cdot e\cdot \overline{W_i(y,v)}$ is well-defined.

By the inductive assumption, $W_i(v)$ and $W_j(u)$ satisfy the property in ($\star$).
In case $W_2(y)=W_i(v)\cdot e\cdot \overline{W_j(y,u)}$, 
every vertex on $W_j(y,u)$ is either on the $y$--$u$ path on $F$ or belongs to some blossom intersecting this path.
This implies that all vertices on $W_j(y,u)$ belong to the new blossom $B$ by the definition of $B$.
Similarly, in case $W_2(y)=W_j(u)\cdot e\cdot \overline{W_i(y,v)}$, all vertices on $W_i(y,v)$ belong to $B$.
Thus, $W_2(y)$ satisfies the property in ($\star$) with the new blossom family $\B\cup \{B\}$.
\end{proof}

We say that a walk $W$ is an {\em extension} of a walk $W'$ if $W=W'+Q$ for some walk $Q$.
\begin{lemma}\label{lem:extension}
For a vertex $z$ on $W_k(x)$, if $z\not\in V(\B)$ holds, then $W_k(x)$ is an extension of $W_1(z)$.
\end{lemma}
\begin{proof}
We show by induction that the statement holds for any search walk $W_k(x)$.
We first suppose that a new search walk is determined in Cases (i) and (ii) of $\scan(v,e)$. Then the new search walk is in the form of $W_i(v)\cdot e\cdot u$. In addition, $i=2$ holds only if $v\in V(\B)$. Since $W_i(v)$ satisfies the property by the inductive assumption, so does the new walk. We next suppose that a new search walk is determined in Case (iii) by $W_2(y)=W_i(v)\cdot e\cdot \overline{W_j(y,u)}$ (resp.,  $W_2(y)=W_j(u)\cdot e\cdot \overline{W_i(y,v)}$). Then, all vertices on $W_j(y,u)$ (resp., $W_i(y,v)$) belong to the new blossom as mentioned in the proof of Lemma~\ref{lem:shrink}. Since $W_i(v)$ (resp., $W_j(u)$) satisfies the property by the inductive assumption, so does $W_2(y)$.
\end{proof}

\begin{lemma}\label{lem:C1}
Any search walk uses each free edge at most once and each vertex at most twice.
\end{lemma}
\begin{proof}
Since the statement clearly holds at the beginning of Step 2, 
it suffices to show that any new walk determined in $\scan(v,e)$ satisfies the properties.

When an update in Case (i) or (ii) of $\scan(v,e)$ is applied, a new search walk is in the form of $W_i(v)\cdot e\cdot u$, where $u\not\in V(\B)$ and $e$ is the stalk of $u$.
By the property ($\star$), $u\not\in V(\B)$ implies that $e$ and $u$ are never used in $W_i(v)$.
Since $W_i(v)$ uses each free edge at most once and each vertex at most twice, 
so does the new walk $W_i(v)\cdot e\cdot u$. 

We next consider an update in Case (\rn{3}). 
By this update, a new search walk $W_2(y)$ is determined for any $y\in B\setminus V(\B)$,
where $\B$ denotes the blossom family before the new blossom $B$ is added.
We suppose that $y$ is an ancestor of $u$ and hence $W_2(y)=W_i(v)\cdot e\cdot \overline{W_j(y,u)}$, because the other case is shown similarly.
We consider two cases depending on whether $y$ is also an ancestor of $v$ or not.

In case $y$ is not an ancestor of $v$, the property ($\star$) and $y\not\in V(\B)$ imply that $W_j(y,u)$ is vertex-disjoint from $W_i(v)$. 
Since each of $W_i(v)$ and $W_j(u)$ uses each free edge at most once and each vertex at most twice by the inductive assumption,
so does $W_2(y)$.

In case $y$ is an ancestor of $v$, let $r$ be the lowest common ancestor of $v$ and $u$. As in the definition of a new blossom $B$, 
let $w$ be the common ancestor furthest from the common root such that the stalk of $w$ is a free edge. Since $y\in B$, it follows that $y$ is on the $w$--$r$ path in $F$.
By the definition of $w$, all edges on this path are labeled.
We claim that the $y$--$r$ path does not intersect any blossom. 
Suppose, to the contrary, that there is a vertex $x$ on the $y$--$r$ path with $x\in B'$ for some blossom $B'\in \B$.
Then $w_{B'}$ is an ancestor of $r$ and the stalk of $w_{B'}$ is free by (B), and hence $w_{B'}$ is an ancestor of $w$ (possibly, $w$ itself). 
Then, all vertices on the $w_{B'}$--$x$ path, including $y$, are contained in $B'$, which contradicts $y\not\in V(\B)$.
Thus, no vertex on the $y$--$r$ path belongs to $V(\B)$. 
By ($\star$), this implies that $W_j(y,u)$ is a concatenation of the $y$--$r$ path and $W_j(r,u)$.
In addition, $r\not\in V(\B)$ and ($\star$) imply that $W_j(r,u)$ is vertex-disjoint from $W_i(v)$ except $r$.
Note that the $y$--$r$ path consists of labeled edges and vertices in $V(F)\setminus V(\B)$, which are used exactly once in $W_i(v)$ by ($\star$).
Therefore, $W_2(y)$ uses each free edge at most once and each vertex at most twice.
\end{proof}

We now complete the proof of admissibility of search walks.

\begin{lemma}[\bf Admissibility]\label{lem:C2}
Any search walk is admissible and uses each vertex at most twice. In addition, $\lambda(W_1(x))\neq \lambda(W_2(x))$ holds for any $x\in V(\B)$.
\end{lemma}
\begin{proof}
What is left after Lemma~\ref{lem:C1} is to show that any search walk satisfies (A2)
and that $\lambda(W_1(x))\neq \lambda(W_2(x))$ holds for any $x\in V(\B)$.
At the beginning of Step 2, all search walks clearly satisfy (A2) and $V(\B)=\emptyset$.
We show by induction that the statement is preserved by any update.

When  Case (i) of $\scan(v,e)$ is applied, the newly determined search walk $W_1(u)=W_i(v)\cdot e\cdot u$ clearly satisfies (A2) and $u\not\in V(\B)$ by the definition of a frontier edge.

When  Case (ii) is applied, the secondary search walk $W_2(u)=W_2(v)\cdot e\cdot u$ is determined.
By the definition of an exterior edge, $e$ is a free edge that is the stalk of $u$, and hence $W_1(u)=W_1(v)\cdot e\cdot u$.
Then, $\lambda(W_i(u))=\lambda(W_i(v))$  for $i=1,2$.
Since $W_2(v)$ satisfies (A2) and $\lambda(W_1(v))\neq\lambda(W_2(v))$ holds by the inductive assumption, 
the same properties are satisfied with $u$.

When Case (\rn{3}) is applied,  a new blossom $B$ is produced and the secondary search walk $W_2(y)$ is determined for any $y\in B\setminus V(\B)$, where $\B$ is the blossom family before $B$ is added.
By the definition of an interior edge,  $Q\coloneqq W_i(v)\cdot e\cdot \overline{W_j(u)}$ satisfies (A2). 
In addition, since $y\not\in V(\B)$, Lemma~\ref{lem:extension} and the property ($\star$) imply that we have $W_j(u)=W_1(y)+W_j(y,u)$
if $y$ is an ancestor of $u$
and $W_i(v)=W_1(y)+W_i(y,v)$ 
if $y$ is an ancestor of $v$.
Thus, $W_2(y)+\overline{W_1(y)}$ coincides with $Q$ or its reverse. Therefore, $\lambda(W_2(y))\neq \lambda(W_1(y))$ follows.
\end{proof}

\subsection{Shortness}\label{sec:shortness}
We next aim at showing the properties of the search walks that leads to the shortness of the output walk of $\search$. Recall that an augmenting walk is short if it has no shortcuts and transfers at each vertex at most twice. The latter condition easily follows from the fact that any search walk uses each vertex at most twice, which is shown in Lemma~\ref{lem:C1} above. 

Therefore, in what follows, we focus on the nonexistence of shortcuts.
Here we use the fact that our algorithm maintains vertices in the two queues $\Phi_1$ and $\Phi_2$, and takes a vertex out of them with priority on the secondary queue $\Phi_2$. 
We first prepare some lemmas.

\begin{lemma}\label{lem:queue1}
At the moment when a vertex $x$ is put into $\Phi_k$ with $W_k(x)$ being determined, every vertex $z$ on $W_k(x)$ except $x$ has been taken out of $\Phi_1$ earlier or has been put into $\Phi_2$ earlier. In particular, the latter holds if $W_k(x)$ is not an extension of $W_1(z)$.
\end{lemma}
\begin{proof}
We show the statement by induction. It obviously holds for the search walks determined in Step~1.
We then show that the statement holds for any new walk determined in $\scan(v,e)$.

In Cases (i) and (ii), the new walk is determined in the form of $W_i(v)\cdot e\cdot u$.
Note that $W_i(v)$ has been determined earlier, and the statement holds for $W_i(v)$ by the inductive assumption. Also, $v$ has been taken out of some queue just before $\scan(v,e)$ is called, and in particular, it is taken out of $\Phi_2$ if $i=2$.
Thus, the statement holds for the new walk $W_i(v)\cdot e\cdot u$.

We next consider Case (\rn{3}), in which $e$ is detected as an interior edge and a new blossom $B$ is produced. 
Let $W_2(y)$ be the secondary walk newly determined for  $y\in B\setminus V(\B)$.
We first assume $W_2(y)\coloneqq  W_i(v)\cdot e\cdot \overline{W_j(y,u)}$. By the same argument as above, the property in the statement is satisfied for every vertex $z$ on $W_i(v)$. By the definition of $B$, every vertex on $\overline{W_j(y,u)}$ is either a descendant of $y$ or belonging to some blossom defined earlier than $B$. Because the algorithm puts vertices on $B\setminus V(\B)$ into the secondary queue $\Phi_2$ in the ascending order, all vertices on $\overline{W_j(y,u)}$ have been put into $\Phi_2$ before $y$. Thus, the statement holds.
We next assume $W_2(y)\coloneqq  W_j(u)\cdot e\cdot \overline{W_i(y,v)}$. By the inductive assumption, every vertex on $W_j(u)$ except $u$ satisfies the property in the statement. Also, by the description of Case (\rn{3}), if $u\not\in V(\B)$, then $u$ is put into $\Phi_2$ before $y$. Similarly to the first case, we can see that all vertices on $\overline{W_i(y,v)}$ are put into $\Phi_2$ before $y$. Thus, the statement holds also in this case.
\end{proof}


\begin{lemma}\label{lem:secondary_queue}
At any moment of the algorithm, if the secondary queue $\Phi_2$ is nonempty, then there exists a blossom $B\in \B$ that contains all the vertices in $\Phi_2$. 
\end{lemma}
\begin{proof}
We show this by induction. At the beginning of the algorithm, $\Phi_2$ is empty, and hence the claim is valid. New vertices are put into $\Phi_2$ only when Case (\rn{2}) or (\rn{3}) of $\scan(v,e)$ is executed for some vertex $v$ and an edge $e$ incident to $v$. 
If $v\in V(\B)$, let $B$ denote the maximal blossom that contains $v$ before $\scan(v,e)$. Otherwise, set $B=\emptyset$. By the inductive assumption, all vertices in $\Phi_2$ must belong to $B$. 
In Case (\rn{2}), the algorithm adds a vertex $u$ to $B$, putting $u$ into $\Phi_2$ at the same time. In Case (\rn{3}), the algorithm produces a new blossom $B'$ that includes $B$. All new vertices the algorithm puts into $\Phi_2$ are contained in $B'$. Thus, in either case, the new vertices $\scan(v,e)$ puts into $\Phi_2$ are in the same maximal blossom containing $v$ after the execution.
\end{proof}

\begin{lemma}\label{lem:queue0}
At the moment when a vertex $x$ is taken out of $\Phi_1$ or $\Phi_2$, if a vertex $z$ belongs to $V(\B)$ and there is no blossom $B\in \B$ with $\{x, z\}\subseteq B$, then $z$ has been taken out of $\Phi_2$ earlier.
\end{lemma}
\begin{proof}
Suppose that a vertex $x$ is taken out of $\Phi_1$ or $\Phi_2$ and a vertex $z$ satisfies $z\in V(\B)$ and there is no blossom $B$ with $\{x,z\}\subseteq B$.
By the algorithm,  $z\in V(\B)$ implies that $z$ has been put into $\Phi_2$ earlier.
In case $x$ is taken out of $\Phi_1$, then $\Phi_2$ is empty.
In case $x$ is taken out of $\Phi_2$, by Lemma~\ref{lem:secondary_queue}, there exists a blossom containing $x$ and all vertices in $\Phi_2$, which implies that $z$ is not in $\Phi_2$ currently.
Thus, in both cases,  $z$ is not in $\Phi_2$, and hence it has been taken out of $\Phi_2$ earlier.
\end{proof}

We are now ready to prove that the primary search walks are free from shortcuts
and secondary search walks may contain shortcuts in a very restricted manner.
For an augmenting walk $Q$ and a shortcut $\S$, we denote by $Q\ast \S$ the augmenting walk obtained from $Q$ by bridging all the pairs in $\S$.
\begin{lemma}[\bf No Shortcut]\label{lem:no-shortcut}
The following two statements hold.
\begin{itemize}
\setlength{\itemsep}{0mm}
\item[\rm (a)] Any primary search walk $W_1(x)$ admits no shortcut.
\item[\rm (b)] If a secondary search walk $W_2(x)$ has a shortcut $\S$, then $\lambda(W_2(x)\ast \S)=\lambda(W_1(x))$.
\end{itemize}
\end{lemma}
\begin{proof}
At the beginning of Step 2, all primary search walks clearly satisfy (a), and no secondary walk is determined.
We show by induction that any new search walk determined in Cases (\rn{1})--(\rn{3}) of $\scan(v,e)$ 
satisfies (a) and (b).


\medskip

{\bf Case (\rn{1}).}
Let $W_1(u)$ be a primary search walk newly determined by $W_1(u)=W_i(v)\cdot e\cdot u$ for $u\not\in V(F)$.
Suppose, to the contrary, that $W_1(u)$ has a shortcut $\S=\{(S_h, S'_h)\}_{h=1}^k$.

We first consider the case where $S'_k$ does not contain the last edge $e$ of $W_1(u)$. In this case, all the $\P$-segments in $\S$ are included in $W_i(v)$, and hence $W_1(u)\ast \S=(W_i(v)\ast \S)\cdot e\cdot u$, which satisfies (A2).
Hence, $\S$ is also a shortcut of $W_i(v)$.
By the inductive assumption, we then have $i=2$ and $\lambda(W_2(v)\ast \S)=\lambda(W_1(v))$ holds.
Therefore, $W_1(v)\cdot e\cdot u$ satisfies (A2), which contradicts $i=2$ by the choice of $i$.

We next consider the case where $S'_k$ contains the last edge $e$ of $W_1(u)$.  In this case, $u$ is the last vertex of $S'_k$.
Let $w$ be the first vertex of $S_k$ and let $P$ be the $T$-path containing $S_k$ and $S'_k$.
In addition, let $Q_w$ be the subwalk of $W_i(v)$ from the first vertex to $w$ (i.e., just before $S_k$) and let $\S'\coloneqq \{(S_h, S'_h)\}_{h=1}^{k-1}=\S\setminus \{(S_k, S'_k)\}$. Then, $W_1(u)\ast \S=(Q_w\ast \S')+P(w, u)$ and it satisfies (A2). 
\begin{claim}\label{claim:short}
If $W_1(w)+P(w,u)$ violates {\em (A2)}, then $Q_w\neq W_1(w)$. 
\end{claim}
\begin{proof}
Suppose that $W_1(w)+P(w,u)$ violates (A2). 
As $(Q_w\ast \S')+P(w, u)$ satisfies (A2), we have $Q_w\ast \S'\neq W_1(w)$.
If $\S'=\emptyset$, this immediately implies $Q_w\neq W_1(w)$.
If $\S'\neq \emptyset$, then $\S'$ is a shortcut of $Q_w$, and hence $Q_w\neq W_1(w)$ follows as $W_1(w)$ has no shortcuts by the inductive assumption. 
\end{proof}

Since $w$ appears on $W_i(v)$ and $v$ was taken out of some queue just before $\scan(v,e)$ is called, Lemma~\ref{lem:queue1} implies that $w$ has been taken out of some queue before $\scan(v,e)$ is called, and in particular, out of $\Phi_2$ if $Q_w\neq W_1(w)$. Note that $Q_w\neq W_1(w)$ holds if $W_1(w)+P(w,u)$ violates (A2) by Claim~\ref{claim:short}. Note also that, when $w$ is taken out of $\Phi_2$, the walk $W_2(w)$ is determined and satisfies $\lambda(W_2(w))\neq \lambda(W_1(w))$.
Therefore, the jumping edge corresponding to $P(w, u)$ should have been detected as a frontier edge before $\scan(v,e)$ is called.
This contradicts the fact that $u\not\in V(F)$ holds when $\scan(v,e)$ is called.

\medskip

{\bf Case (\rn{2}).}
In this case,  $W_2(u)=W_2(v)\cdot e\cdot u$ is newly determined for $u\not\in V(\B)$.
As $e$ is free, any shortcut of $W_2(u)$ is also a shortcut of $W_2(v)$.
Furthermore, as $v$ is the parent of $u$, we have $W_1(u)=W_1(v)\cdot e\cdot u$,
and hence $\lambda(W_1(v))=\lambda(W_1(u))$.
Since $W_2(v)$ satisfies (b) by the inductive assumption, so does $W_2(u)$.

\medskip

{\bf Case (\rn{3}).} An interior edge $e$ is detected and a new blossom $B$ is produced.
A new search walk is defined by $W_2(y)\coloneqq W_i(v)\cdot e\cdot \overline{W_j(y,u)}$ or
$W_2(y)\coloneqq W_j(u)\cdot e\cdot \overline{W_i(y,v)}$ for each $y\in B\setminus  V(\B)$, where $\B$ denotes the blossom family just before $B$ is added.
We assume $W_2(y)\coloneqq W_i(v)\cdot e\cdot \overline{W_j(y,u)}$ because the other case can be shown similarly.
Consider a walk $Q\coloneqq W_i(v)\cdot e\cdot \overline{W_j(u)}$.
Since $y\notin V(\B)$,  Lemma~\ref{lem:extension} implies
$W_j(u)=W_1(y)+W_j(y,u)$. 
Then, $Q= W_i(v)\cdot e\cdot \overline{W_j(y,u)} +\overline{W_1(y)}= W_2(y)+ \overline{W_1(y)}$.

To show (b), suppose to the contrary that $W_2(y)$ has a shortcut $\S=\{(S_h, S'_h)\}_{h=1}^k$
such that $\lambda(W_2(y)\ast \S)\neq \lambda(W_1(y))$.
Then $W_2(y)\ast \S+\overline{W_1(y)}$ satisfies (A2).
Since $W_2(y)\ast \S+\overline{W_1(y)}=Q\ast \S$,
the walk $Q\ast \S$ satisfies (A2). Consider 
\begin{align*}
\S^v&\coloneqq \set{(S_h, S'_h)\in \S|\text{$S_h$ and $S'_h$ are included in $W_i(v)$}},\\
\S^u&\coloneqq\set{(S_h, S'_h)\in \S|\text{$S_h$ and $S'_h$ are included in $\overline{W_j(u)}$}}.
\end{align*}
Since $S_1, S'_1, \dots, S_k, S'_k$ appear in this order on $W_2(y)=W_i(v)\cdot e\cdot \overline{W_j(y,u)}$, there is at most one index $h^*$ satisfying $(S_{h^*}, S'_{h^*})\not\in \S^v\cup \S^u$.
Thus, either of the following two holds (see Figure~\ref{fig:no-shortcut}): 
\begin{enumerate}
\item ~$\S^v\cup \S^u=\S$,
\item ~$\S\setminus (\S^v\cup \S^u)=\{(S_{h^*}, S'_{h^*})\}$ for some $h^*\in \{1,2,\dots,k\}$.
\end{enumerate}
\begin{figure}[t]
	\begin{center}
		\includegraphics[width=0.6\hsize]{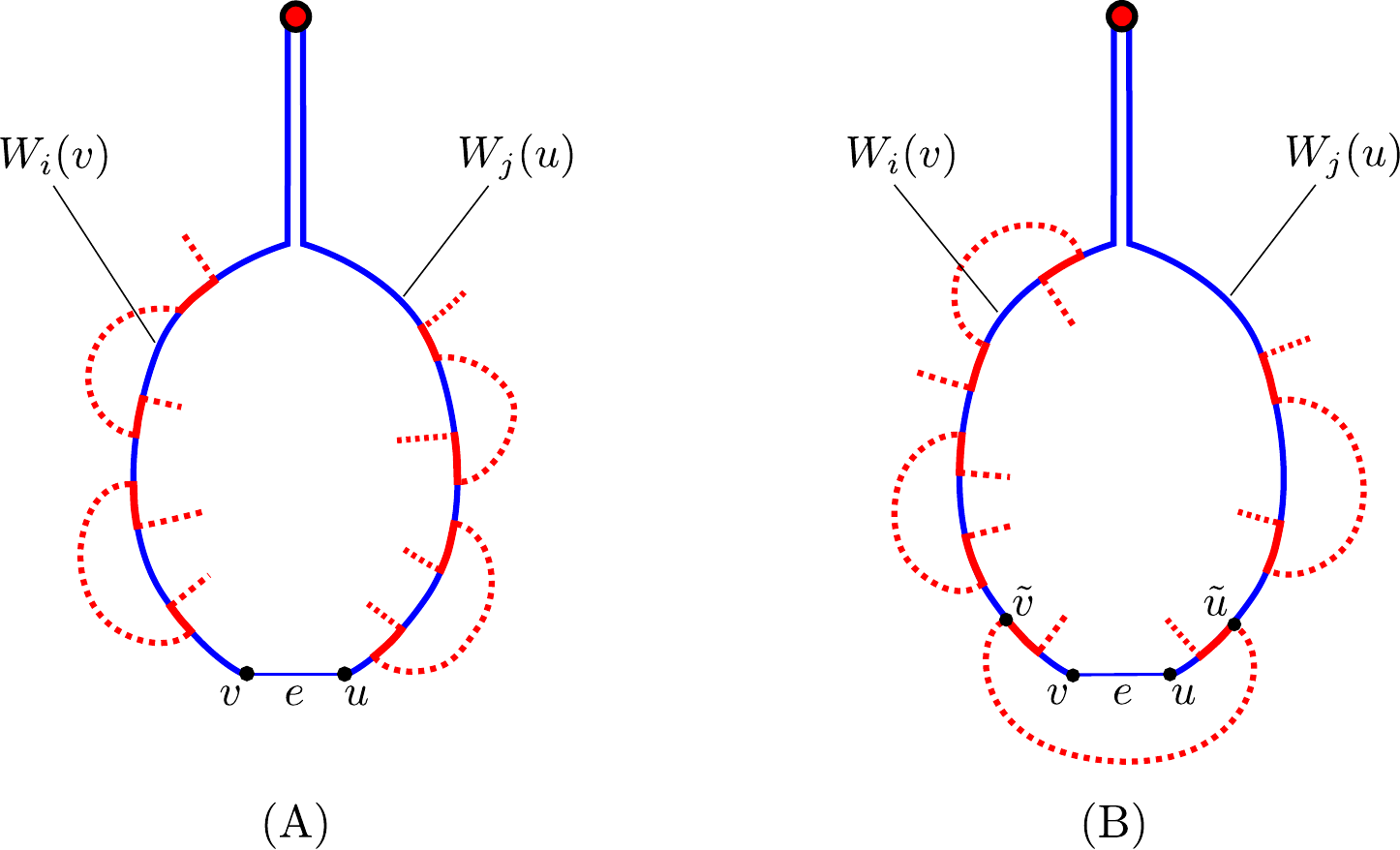}
	\end{center}
	\vspace{-3mm}
	\caption{
		\small\baselineskip=11pt
The figures (A) and (B) respectively show images of the cases $\S^v\cup \S^u=\S$ and $\S^v\cup \S^u\neq \S$. Blue lines represent the search walks $W_i(v)$ and $W_j(u)$. Thick red lines represent the $\P$-segments constituting of the shortcut $\S$. Dashed red lines represent the $T$-paths from which those $\P$-segments come.}
	\label{fig:no-shortcut}
\end{figure}

In the first case, we have $Q\ast \S
=(W_i(v)\ast \S_v)\cdot e\cdot (\overline{W_j(u)}\ast \S_u)$.
Since $Q\ast \S$ satisfies (A2), so do 
$W_i(v)\ast \S_v$ and $\overline{W_j(u)}\ast \S_u$.
By the inductive assumption on $W_i(v)$ (resp.,  $W_j(u)$), we have $\S^v=\emptyset$ (resp, $\S^u=\emptyset$) if $i=1$ (resp., $j=1$) and the statement (b) holds if $i=2$ (resp., $j=2$).
In any case, we have $\lambda(W_i(v)\ast \S^v)=\lambda(W_1(v))$ and $\lambda(W_j(u)\ast \overline{\S^u})=\lambda(W_1(u))$, where $\overline{\S^u}$ is the collection of pairs of $\P$-segments in $W_j(u)$ corresponding to $\S^u$. 
Since $Q\ast \S$ satisfies (A2), then $W_1(v)\cdot e\cdot \overline{W_1(u)}$ satisfies (A2). By the choice of $i$ and $j$, we have $i=j=1$, which implies by the inductive assumption that $\S^v=\S^u=\emptyset$ holds while $\S$ is nonempty, a contradiction.

\medskip
In the rest, we consider the second case (see  Figure~\ref{fig:no-shortcut} (B)). Let $\tilde{v}$ be the first vertex of $S_{h^*}$ and  $\tilde{u}$ be the last vertex of $S'_{h^*}$. 
Then, $\tilde{v}$ and $\tilde{u}$ appear on $W_i(v)$ and $W_j(u)$, respectively.
Let $Q_{\tilde{v}}$ be the subwalk of $W_i(v)$ from the first vertex to $\tilde{v}$ (i.e., just before $S_{h^*}$)
and let $Q_{\tilde{u}}$ be the subwalk of $W_j(u)$ from the first vertex to $\tilde{u}$ (i.e., just before $\overline{S'_{h^*}}$).
Then, $Q\ast \S= (Q_{\tilde{v}}\ast \S^v)+P(\tilde{v}, \tilde{u})+(\overline{Q_{\tilde{u}}}\ast \S^u)$, which satisfies (A2), where $P$ denotes the $T$-path containing $S_{h^*}$ and $S'_{h^*}$.
Since $(Q_{\tilde{v}}\ast \S^v)+P(\tilde{v}, \tilde{u})$ and $(Q_{\tilde{u}}\ast \overline{\S^u})+P(\tilde{u}, \tilde{v})$ satisfy (A2), similarly to Claim~\ref{claim:short}, one can show the following two statements.
\begin{itemize}
\item[($\dag$)] If $W_1(\tilde{v})+P(\tilde{v},\tilde{u})$ violates (A2), then $Q_{\tilde{v}}\neq W_1(\tilde{v})$.
\item[($\ddag)$] If $W_1(\tilde{u})+P(\tilde{u},\tilde{v})$ violates (A2), then $Q_{\tilde{u}}\neq W_1(\tilde{u})$.
\end{itemize}

\begin{claim}\label{claim:short1}
At the moment when the procedure $\scan(v,e)$ is called, there is no blossom $B'\in \B$ that satisfies \mbox{$\{\tilde{v}, \tilde{u}\}\subseteq B'$} or $\{v, \tilde{u}\}\subseteq B'$.
\end{claim}
\begin{proof}
Suppose, to the contrary, there exists a blossom $B'\in \B$ satisfying $\{\tilde{v}, \tilde{u}\}\subseteq B'$ or $\{v, \tilde{u}\}\subseteq B'$. Here we assume the former because the latter case can be shown similarly. 
In the shrunk forest $F/\B$, the projections of $\tilde{v}$ and $\tilde{u}$ must coincide. Since $F[B']$ forms a subtree of $F$, it follows from Lemma~\ref{lem:shrink} that the calyx $w_{B'}$ of $B'$ is a common ancestor of $v$ and $u$. Since the stalk of $w_{B'}$ is a free edge, this implies that the calyx $w_B$ of the new blossom $B$ is a descendant of $w_{B'}$ (possibly $w_B=w_{B'}$), and hence $y\in B\setminus V(\B)$ is also a descendant of $w_{B'}$. Since $\tilde{u}$ is on $W_j(y,u)$, however, $y$ is an ancestor of $\tilde{u}\in B'$. Therefore, we have $y\in B'$, which contradicts $y\not\in V(\B)$.
\end{proof}

Let us denote by $\tilde{e}$ the jumping edge corresponding to $P(\tilde{v}, \tilde{u})$.
\begin{claim}\label{claim:short2}
We have $\tilde{e}\not\in E^*(F)$. 
\end{claim}
\begin{proof}
Suppose, to the contrary, $\tilde{e}\in E^*(F)$. Then $\tilde{e}$ is the stalk of $\tilde{v}$ or $\tilde{u}$. Without loss of generality, let $\tilde{e}$ be the stalk of $\tilde{v}$. Then, the last edge of $W_1(\tilde{v})$ is $\tilde{e}$, and hence $W_1(\tilde{v})+P(\tilde{v}, \tilde{u})$ violates (A2), which implies $Q_{\tilde{v}}\neq W_1(\tilde{v})$ by ($\dag$). Lemma~\ref{lem:extension} then implies $\tilde{v}\in V(\B)$, i.e.,  there exists a blossom $B'$ with $\tilde{v}\in B'$. 
Since the stalk $\tilde{e}$ of $\tilde{v}$ is labeled, $\tilde{v}$ is not the calyx of $B'$, and hence its parent $\tilde{u}$ also belongs to $B'$. 
We then have $\{\tilde{v}, \tilde{u}\}\subseteq B'$, which contradicts Claim~\ref{claim:short1}.
\end{proof}

We now show that, before $W_2(y)$ is determined in $\scan(v,e)$, the algorithm must have detected $\tilde{e}$ as an interior edge, which contradicts Claim~\ref{claim:short1} and completes the proof.

Since $v$ was taken out of some queue before $\scan(v,e)$ is called and $\tilde{v}$ is on $W_i(v)$,  Lemma~\ref{lem:queue1} implies that $\tilde{v}$ was taken out of some queue before $\scan(v,e)$ is called. In particular, $\tilde{v}$ was taken out of $\Phi_2$ if $Q_{\tilde{v}}\neq W_1(\tilde{v})$. This together with ($\dag$) implies that
$W_{i^*}(\tilde{v})+P(\tilde{v},\tilde{u})$ satisfies (A2) for some $i^*$ at the moment when $\tilde{v}$ was taken out of that queue.
Since $\tilde{e}\not\in E^*(F)$ by Claim~\ref{claim:short2}, then the edge $\tilde{e}$ must have been detected as an interior edge except in the case where $W_1(\tilde{u})+P(\tilde{u}, \tilde{v})$ violates (A2) and $W_2(\tilde{u})$ is not determined at that moment.
In such an exceptional case, ($\ddag$) implies $Q_{\tilde{u}}\neq W_2(\tilde{u})$, from which Lemma~\ref{lem:extension} implies that $\tilde{u}\in V(\B)$ holds at the moment $W_j(u)$ is determined and afterwords.
By Claim~\ref{claim:short1} and  Lemma~\ref{lem:queue0}, $\tilde{u}$ has been taken out of $\Phi_2$ before $\scan(v,e)$ is called, and at that moment, $\tilde{e}$ must have been detected as an interior edge. 
Thus, we have shown that in any case $\tilde{e}$ has been detected as an interior edge before $\scan(v,e)$ is called, a contradiction.
\end{proof}

So far, we have seen properties of the search walks determined in the algorithm. We now show the shortness of the walk returned by the algorithm.
\begin{theorem}\label{thm:search-main}
If the algorithm $\search$ terminates with returning a walk, then the output is a short augmenting walk in $\G^*(\P)$, which corresponds to a short augmenting walk in $\G(\P)$.
\end{theorem}
\begin{proof}
Suppose that the algorithm returns a walk $Q\coloneqq W_i(v)\cdot e\cdot \overline{W_j(u)}$ at Step (2-2).
Let $t$ and $f$ be the last vertex and edge of $Q$, respectively, and
delete $t$ and $f$ to obtain $Q'$, i.e., $Q=Q' \cdot f\cdot t$.
By applying the arguments used in the proofs of Lemmas~\ref{lem:C2} and \ref{lem:no-shortcut} to analyze Case (\rn{3}) of $\scan(v,e)$, we see that $Q'$ is an admissible walk that uses each vertex at most twice and that, for any shortcut $\S$ of $Q'$, we have $\lambda(Q'\ast \S)=\lambda(W_1(z))$, where $z$ is the last vertex of $Q'$.
Since $f$ is a free edge, we have $\lambda(W_1(z))=t$, and hence
$\lambda(Q'\ast \S)=\lambda(W_1(z))$  implies that 
$Q\ast \S=(Q'\ast \S)\cdot f\cdot t$ violates (A2).
Therefore, $Q$ has no shortcuts.
In addition, as $Q$ uses each vertex at most twice, it transfers at each vertex at most twice. 
Thus, $Q$ is a short augmenting walk in $\G^*(\P)$. By Lemma~\ref{lem:saw_2}, this output determines a short augmenting walk in $\G(\P)$. 
\end{proof}
As mentioned before, we will show in Section~\ref{sec:linear-scan} that one can improve the algorithm to run with $O(|E|)$ calls of $\scan$ by appropriately skipping redundant scans using pointers on $T$-paths.

\subsection{Optimality}
In this section, we analyze the situation when the algorithm terminates without returning a walk. We intend to show that in such a case the collection $\P$ consists of the maximum number of edge-disjoint $T$-paths.

\begin{theorem}\label{thm:tightness}
If the algorithm terminates without returning a walk,
then $\P$ consists of the maximum number of edge-disjoint $T$-paths in $G=(V,E)$.
\end{theorem}
\begin{proof}
Consider $F$ and $\B$ at the end of the algorithm. 
Recall that  the forest $F$ is a subgraph of $G^*=(V, E^*)$.
For each $s\in T$, define $X_s$ by
\[X_s=\set{v\mid v\in V(F)\setminus V(\B),~ \lambda(W_1(v))=s}.\]
As $s\in X_s$ holds for each $s\in T$, the family $\X=\{X_s\}_{s\in T}$ forms a $T$-subpartition.
Recall that the algorithm terminates when the two queues become empty and that, whenever it takes out a vertex, all incident edges in $E^*\cup L$ are scanned. Therefore, there is neither a frontier, exterior, nor interior edge with respect to $F$ and $\B$.
We show the following three claims. 
\begin{enumerate}
\setlength{\itemsep}{1mm}
\item No free edge connects $X_s$ and $X_t$ with distinct $s,t\in T$.
\item If $e\in E^*$ is a labeled edge with $\partial e=\{u,v\}$, $u\in X_s$, and $v\not\in X_s$, then $\sigma_{E^*}(e,u)=s$.
\item The vertex set $K$ of any connected component of $G^*\setminus \X$ forms a subtree of $F$ or is disjoint from $V(F)$.
In the former case, exactly one free edge connects $\bigcup_{s\in T}X_s$ and $K$.
In the latter case, there is no such free edge.
\end{enumerate}

The first claim is obvious by the nonexistence of an interior edge.

To show the second claim, suppose to the contrary that $\sigma_{E^*}(e,u)\neq s$.
In case $\sigma_{E^*}(e,v)= s$, the walk $W_1(u)\cdot e\cdot v$ satisfies (A2) and its last symbol is $s$.
As $v\not\in X_s$, this means that $e$ is a frontier or interior edge, a contradiction.
In the other case, i.e., when $\sigma_{E^*}(e,v)\neq s$, the vertex $u\in X_s$ has a selfloop with two symbols distinct from $s$,
which is an interior edge, a contradiction.

For the third claim, note that every vertex $v$ in $G^*\setminus \X$
is either $v\in V(\B)$ or $v\in V\setminus V(F)$.
By the nonexistence of a frontier edge, there is no edge connecting $V(\B)$ and $V\setminus V(F)$.
Therefore, the vertex set $K$ of each connected component is included in $V(\B)$ or $V\setminus V(F)$.

We first consider the case $K\subseteq V(\B)$.
Since each blossom in $\B$ is a subtree of $F$ by the property (B) in Lemma~\ref{lem:subtree}, it is connected in $G^*\setminus \X$.
Then, $K$ is the union of one or more maximal members of $\B$.
Suppose, to the contrary, that $K$ is not a subtree of $F$, i.e., the restriction of $F$ to $K$ is not connected.
As $K$ is connected in $G^*\setminus \X$, there is an edge $e\in E^*\setminus E^*(F)$ connecting two disjoint blossoms $B, B'\subseteq K$. Then, there is no blossom containing the two end-vertices of $e$ and each of them has two search walks,
which means that $e$ is an interior edge, a contradiction.
Thus, $K$ forms a subtree of $F$ that consists of one or more blossoms in $\B$.
The root of $K$ is the calyx of some of those blossoms, and hence its stalk is a free edge by (B).
Also, any other edge connecting $K$ and $V(F)\setminus V(\B)$
should be a labeled edge, because otherwise it is an exterior or interior edge.
In case $K\subseteq V\setminus V(F)$,
there is no free edge connecting $K$ and $\bigcup_{s\in T}X_s$,
because otherwise we have a frontier edge.

Using the above three claims, we show
$|\P|=\frac{1}{2}\left[\sum_{s\in T}d(X_s)-\odd(G\setminus\X)\right]$,
which guarantees the maximality of the number of $T$-paths by Lemma~\ref{lem:Mader}.
In the summation $\sum_{s\in T}d(X_s)$, every $T$-path  is counted exactly twice by the second claim.
This implies that labeled edges are counted $2|\P|$ times in total.
Also, free edges are counted $\odd(G\setminus\X)$ times by the first and third claims.
Thus, we have $\sum_{s\in T}d(X_s)=2|\P|+\odd(G\setminus\nobreak\X)$, which is equivalent to the required equality.
\end{proof}

\section{Efficient Implementation}\label{sec:implementation}
Combining results in Sections~\ref{sec:augmentation} and \ref{sec:search},
we obtain an algorithm to compute the maximum number of edge-disjoint $T$-paths. The algorithm starts with $\P\coloneqq \emptyset$. It then repeatedly  finds a short augmenting walk $Q$ in $\G(\P)$ and updates $\P$ with $Q$ by the switching operation. The algorithm terminates when no augmenting walk is found.

By (A3) and Lemma~\ref{lem:property_of_shortness}, the length of a short augmenting walk $Q$ is $O(|E|)$, and hence the switching operation to obtain $\P\triangle Q$ requires $O(|E|)$ time.
In Section~\ref{sec:search}, we have described the algorithm $\search$ in a manner suitable for the correctness proof.
In this section, we present its efficient implementation, which runs in $O(|E|)$ time.

\subsection{Reducing the Number of Scans}\label{sec:linear-scan}
In the algorithm $\search$, each vertex is put into each of the two queues $\Phi_1$ and $\Phi_2$ at most once, and an edge is scanned only when one of its end-vertices is taken out of some queue. Therefore, each edge $e\in E^*\cup L$ is scanned at most four times. Since $|E^*|=O(|E|^2)$ and $|L|=O(|E|)$, this means that the procedure $\scan$ is called $O(|E|^2)$ times in the algorithm. 
Here we show that we can reduce the number of scans to $O(|E|)$ by skipping redundant scans. 

To this end, we introduce pointers that move on the $T$-paths in $\P$.
For each $P\in \P$ with terminals $s$ and $t$, 
the algorithm maintains two pointers $\eta_{P}^{st}$ and $\eta_{P}^{ts}$, each of which points a vertex on $P$.
Initially, $\eta_{P}^{st}$ (resp., $\eta_{P}^{ts}$) points $t$ (resp., $s$), and then it repeatedly gets closer to $s$ (resp., $t$)
during the algorithm.

The new version of $\search$ is obtained by adding the initialization of pointers mentioned above to Step~1 and modifying Step~(2-2) as follows. In the original $\search$, Step~(2-2) calls $\scan(v,e)$ for every edge $e\in E^*\cup L$ incident to the dequeued vertex $v$. In the new version, Step~(2-2) is replaced with the following one.
\begin{enumerate}
\setlength{\leftskip}{10mm}
\item[(2-$2^\star$)] For each free edge $e$ incident to $v$, call $\scan(v,e)$.\\
For each $P\in \P$ containing $v$, let $s$ and $t$ be its end-terminals and do the following. 
\begin{itemize}
\setlength{\itemsep}{0mm}
\setlength{\leftskip}{5mm}
\item If $\lambda(W_i(v))\neq s$ for some $i\in \{1,2\}$ and $v$ is an intermediate vertex of $P(s, \eta_{P}^{st})$, then apply the following.
\begin{itemize}
\setlength{\itemsep}{0mm}
\setlength{\leftskip}{2mm}
    \item For each vertex $u$ on $P(v, \eta_{P}^{st})$,
let $e\in E^{*}\cup L$ be the edge corresponding to the $st$-directed subpath/selfloop $P(v,u)$ and call $\scan(v,e)$. 
\item Update $\eta_{P}^{st}\gets v$.
\end{itemize}
\item Do the same with $s$ and $t$ being all interchanged.
\end{itemize}
\end{enumerate}	

Let us call this version of the algorithm $\searchfast$.
Consider the moment when a vertex $v$ is taken out of some queue and let $P$ be a $T$-path that contains $v$. 
In the original $\search$, all edges $e\in E^*\cup L$ coming from $P$ are scanned. In contrast, 
$\searchfast$ skips to call $\scan(v, e)$ for edges $e$ that do not satisfy the conditions described in Step (2-$2^\star$) above. 

Observe that $\search$ has flexibility 
in the order of edges to be scanned in Step~(2-2). 
We intend to show that, in some execution of $\search$, the edges skipped in $\searchfast$ are never detected as frontier nor interior edges.
This implies that $\searchfast$ indeed simulates some execution of $\search$.
To this end, we prepare some observations on the pointers.
\begin{lemma}\label{lem:scan1}
In $\searchfast$, the following two statements hold for any $P\in \P$ with terminals $s$ and $t$.
\end{lemma}
\begin{enumerate}
	\item[(a)]  $\lambda(W_1(\eta_{P}^{st}))\neq s$ or $\eta^{st}_{P}\in V(\B)$.
	\item[(b)] For any internal vertex $u$ of $P(\eta_{P}^{st}, t)$, we have $\lambda(W_1(u))=t$ or $\{u,\eta_{P}^{st}\}\subseteq B$ for some $B\in\B$.
\end{enumerate}
\begin{proof}
By the description of Step~(2-$2^\star$), the pointer $\eta^{st}_P$ is updated to a vertex $v$ only if $\lambda(W_i(v))\neq s$ for some $i\in \{1,2\}$, which implies (a). To see (b), consider the update of $\eta_P^{st}$ that makes $u$ an internal vertex of $P(\eta_P^{st},t)$ for the first time. Let $v$ be the vertex such that $\eta_P^{st}=v$ as a result of this update. Immediately before the update, $\scan(v,e)$ is applied to a labeled edge $e\in E^*$ with $\partial e=\{v,u\}$. If $u\notin V(F)$ at the moment, $e$ is a frontier edge, and hence $\lambda(W_1(u))=t$ after $\scan(v,e)$. If $u\in V(F)$, $\lambda(W_1(u))\neq t$, and there is no blossom that includes $\{v,u\}$ at the moment, then $e$ is an interior edge and hence $\scan(v,e)$ produces a new blossom that includes $\{v,u\}$. Thus, in any case, $\lambda(W_1(u))=t$ or $\{u,\eta_P^{st}\}\subseteq B$ for some $B\in\B$ immediately after the update. 
It then suffices to show that, once $\{u,\eta_{P}^{st}\}$ gets included in a blossom, this situation is preserved until the end of the algorithm.
Suppose that $\eta_{P}^{st}$ is updated from $v_{1}$ to $v_{2}$ and that $\{u,v_1\}\subseteq B$ holds for some $B\in\B$ immediately before the update, which implies $v_1\in V(\B)$.
In the case where $\{v_1,v_2\}\subseteq B'$ for some blossom $B'\in\B$, there should be some $B''\in \B$ with $\{u,v_1,v_2\}\subseteq B''$ as $\B$ is a laminar family. 
In the other case, the labeled edge corresponding to $P(v_2, v_1)$ is detected as an interior edge,
which produces a blossom including $\{u,v_1,v_2\}$. Thus, the situation is preserved.
\end{proof}

\begin{lemma}\label{lem:scan2}
In $\searchfast$, for any $T$-path $P\in \P$ and its terminals $s$ and $t$,
if $s, \eta_{P}^{st}, \eta^{ts}_{P}, t$ appear in this order on $P$ and $\eta_{P}^{st}\neq \eta^{ts}_{P}$,
then there exists a blossom containing all vertices on  $P(\eta_{P}^{st}, \eta^{ts}_{P})$.
\end{lemma}
\begin{proof}
	At the beginning of the algorithm, we have $\eta^{ts}_{P}=s$ and $\eta_{P}^{st}=t$.
	Consider the moment when $\eta_{P}^{st}$ becomes closer to $s$ than $\eta^{ts}_{P}$ for the first time.
	It happens when $\eta_{P}^{st}$  or $\eta^{ts}_{P}$ is updated.
	Without loss of generality, suppose that this happens when $\eta_{P}^{st}$ is updated from $v_1$ to $v_2$.
	That is, before the update,  $s, v_2, \eta^{ts}_{P}, v_1=\eta_{P}^{st}, t$ appear in this order on $P$ and $v_2\neq \eta^{ts}_{P}$.
	The algorithm performs the update after taking $v_2$ out of some queue and scanning edges incident to $v_2$, including
	the labeled edge $e$ corresponding to the $st$-directed subpath $P(v_2, \eta^{ts}_{P})$. By Lemma~\ref{lem:scan1}\,(a), we have $\lambda(W_1(\eta_{P}^{ts}))\neq t$ or $\eta^{ts}_{P}\in V(\B)$, 
	and hence $e$ is detected as an interior edge unless there is a blossom containing both $v_2$ and $\eta^{ts}_{P}$.
	After $\eta_{P}^{st}$ is updated to $v_2$, there exists a blossom $B^*$ with $\{\eta_{P}^{st}, \eta^{ts}_{P}\}\subseteq B^*$.
	Also, by Lemma~\ref{lem:scan1}\,(b), for each vertex $u$ on $P$ between $\eta_{P}^{st}$ and $\eta^{ts}_{P}$, there exists a blossom $B\in \B$ with $\{u, \eta_{P}^{st}\}\subseteq B$ or $\{u, \eta^{ts}_{P}\}\subseteq B$.
	(To see this, note that $\lambda(W_1(u))=t$ and $\lambda(W_1(u))=s$ cannot hold simultaneously).
	Since we have $\{\eta_{P}^{st}, \eta^{ts}_{P}\}\subseteq B^*\in\B$ and $\B$ is a laminar family, this implies that there exists a blossom containing all vertices on $P(\eta_{P}^{st}, \eta^{ts}_{P})$. For those vertices $u$, as shown in the proof of Lemma~\ref{lem:scan1}, it is preserved through the algorithm that there is a blossom including $\{u, \eta_{P}^{st},\eta^{ts}_{P}\}$. In addition, every vertex that is added to $P(\eta_{P}^{st}, \eta^{ts}_{P})$ by the later updates of $\eta^{st}_P$ and $\eta^{ts}_P$ also satisfies that property by Lemma~\ref{lem:scan1}\,(b). 
\end{proof}

We are now ready to show that $\searchfast$ simulates some execution of $\search$.
\begin{lemma}\label{lem:scan-linear-correct}
The output of $\searchfast$ is the same as that of some execution of $\search$.
\end{lemma}
\begin{proof}
Consider a possible execution of $\search$ such that, in Step~(2-2), the algorithm first scans edges according to Step~(2-$2^\star$) and then scans remaining edges in $E^*\cup L$ incident to $v$.
To see the statement, it suffices to show that in the latter part no edge is detected as a frontier edge nor interior edge. (Note that those remaining edges are labeled and cannot be exterior.)

Suppose that an edge $e\in E^*\cup L$ is incident to $v$ but not scanned in Step~(2-$2^\star$) when $v$ is taken out of some queue.  Set $s\coloneqq \sigma_{E^*}(e,v)$ and $t\coloneqq \sigma_{E^*}(e,u)$ (or $st\coloneqq \sigma_L(e)$ if $e\in L$). Let $P\in \P$ be the $T$-path from which $e$ comes. As $e$ is not scanned, one of the following holds just before Step~(2-$2^\star$) is applied: (i)  $\lambda(W_1(v))=s$ and $v\not\in V(\B)$, (ii) $v$ is on $P(\eta^{st}_P, t)$, or (iii) $u$ is an internal vertex of $P(\eta_P^{st}, t)$. In case (i), clearly $e$ is neither frontier nor interior. In case (ii), 
since $P(v,u)$ is $st$-directed, the condition (iii) also holds unless $e$ is the selfloop at $\eta^{st}_P$, where such a selfloop cannot be interior as it had been scanned when $\eta^{st}_P$ was updated last time. 

Therefore, it suffices to consider the case where (i) fails and (iii) holds.
By Lemma~\ref{lem:scan1}\,(b), $e$ is not frontier. Suppose, to the contrary, that  $e$ is interior just after Step (2-$2^\star$) is applied to $v$. By Lemma~\ref{lem:scan1}\,(b) and the definition of interior edges, we have $u\in V(\B)$, $v\neq u$, and there is no blossom $B\in \B$ with $\{v, u\}\subseteq B$. By Lemma~\ref{lem:queue0}, this implies that $u$ has been  taken out of $\Phi_2$ before the dequeue of $v$, and hence $\eta^{ts}_{P}$ currently points $u$ or a vertex closer to $t$.
Also, since (i) fails, after Step (2-$2^\star$) is applied to $v$, the pointer $\eta^{st}_P$ points $v$ or a vertex closer to $s$. As we have $v\neq u$, these imply that $s, \eta^{st}_P, v, u, \eta^{ts}_P, t$ appear in this order on $P$ and $\eta^{st}_P\neq \eta^{ts}_P$. By Lemma~\ref{lem:scan2}, then there exists a blossom $B$ with $\{v, u\}\subseteq B$, a contradiction.
\end{proof}

The following fact is  observed from the description of the algorithm $\searchfast$.  
\begin{lemma}\label{lem:scan_linear}
The procedure $\scan$ is called $O(|E|)$ times in $\searchfast$.
\end{lemma}
\begin{proof}
Since each vertex is put into each of the two queues at most once, each free edge is scanned at most twice from each end-vertex, and hence at most four times in total. 
For any $T$-path $P\in \P$, let $E(P)$ be the set of edges on $P$, and let $s$ and $t$ be the end-terminals of $P$.
We say that an edge $e\in E^*\cup L$ coming from $P$ is scanned in $st$-direction if $\scan(v,e)$ is called with either $\sigma_{E^*}(e,v)=s$ or $\sigma_{L}(e)=st$. 
Note that, at Step (2-$2^\star$) of $\searchfast$, the number of edges scanned in $st$-direction is the number of vertices in $P(v, \eta_{P}^{st})$,
and after those scans, $\eta_{P}^{st}$ is updated to point $v$.
As the pointer $\eta_P^{st}$ moves monotonically towards $s$, it follows that edges coming from $P$ are scanned in $st$-direction $O(|E(P)|)$ times. 
Thus, the total number of scans during the algorithm is $O(|E|)$.
\end{proof}

\subsection{Implementation Details}
We have shown in Lemma~\ref{lem:scan_linear} that our search algorithm runs with $O(|E|)$ calls of the procedure $\scan$. In this section, we provide other details to attain an $O(|E|)$ time implementation of the search algorithm $\searchfast$.

Our implementation does not store search walks explicitly. Instead, we store $\sym(u)\coloneqq \lambda(W_1(u))$ for each $u\in V(F)$.
Recall that, in $\searchfast$,  a primary search walk $W_1(u)$ is defined when the stalk $e$ of $u$ is detected as a frontier edge. Instead, our implementation keeps $\sym(u)$, which can be determined by $\sym(u)\coloneqq \sigma_{E^*}(e,u)$ if $e$ is labeled and $\sym(u)\coloneqq \sym(v)$ if $e$ is a free edge with $\partial e=\{v,u\}$.
For any edge $e$ with $\partial e=\{v,u\}$, we can easily discern from $\sym(v)$ whether $W_1(v)\cdot e\cdot u$ satisfies (A2) or not (where $W_1(v)$ is assumed to satisfy (A2)). 
Furthermore, when $W_1(v)\cdot e\cdot u$ violates (A2), the condition $\lambda(W_1(v))\neq \lambda(W_2(v))$ implies that $W_2(v)\cdot e\cdot u$ must satisfy (A2).
Thus, it is sufficient to store $\sym(v)$ to detect frontier and interior edges and to specify the indices $i$ and $j$ in $\scan$.
In the following, we describe how to manage blossoms and how to reconstruct the output walk to be returned by $\searchfast$.

\paragraph{Storing blossoms.}
Our implementation stores the blossom family $\B$ as a collection of the indices (names) of the blossoms. 
An edge $e_B$ and a vertex $w_B$ are associated with each $B\in \B$.
For each vertex $v\in V$, we store $\minimal(v)\in \B\cup\{\varnothing\}$, which specifies the minimal blossom containing $v$ if $v\in V(\B)$ and $\varnothing$ otherwise.

We also store a family $\B^*$, which is a partition of $V$ consisting of the maximal members of $\B$ and singletons $\{x\}$ of all $x\in V\setminus V(\B)$.
Each member of $\B^*$ is identified by its representative element,
which is the calyx $w_B$ for a blossom $B\in \B^*$ and is $x$ for a singleton $\{x\}\in \B^*$.
Since each member of $\B^*$ forms a subtree of the forest $F$, the family $\B^*$ can be maintained by using the data structure {\em incremental tree set union} \cite{GT85}, which enables us to apply the following operations so that the total time complexity is linear to the total number of queries.
\begin{itemize}
\item $\grow(v,u)$: Add a new vertex $u$ to $F$ with $v$ being the parent of $u$.
\item $\find(v)$: Return the representative element of the member of $\B^*$ containing $v$.
\item $\link(v)$: Unite the member in $\B^*$ containing $v$ and the member containing the parent of $v$.
\end{itemize}
In our implementation, we call $\grow$ when a frontier edge is detected, and hence $\grow$ is called at most $|V|$ times.
When an exterior edge is detected, we call $\link$ just once. 
When an interior edge is detected, we create a new blossom and update $\B^*$ calling $\link$ and $\find$ multiple times, as described below. 

\paragraph{Creating a new blossom.}
When an edge $e$ is detected as an interior edge, 
a new blossom $B$ with $e_B=e$ is produced or a walk is returned by $\searchfast$.
In the former case, we have to set $\minimal(y)\coloneqq B$ for each $y\in B\setminus V(\B)$ and have to update $\B^*$.

For this purpose, we assign to each vertex $u\in V(F)$ an integer  $\seq(u)$ which reflects the order that $u$ is added to $F$. At the beginning of $\searchfast$, we set $\seq(t)=0$ for every $t\in T$ and set a counter ${\sf cnt}\gets 1$. Whenever a new vertex $u$ is added to $F$, we set $\seq(u)\coloneqq {\sf cnt}$ and increment ${\sf cnt}$. 

We now introduce an operation $\ascend(x)$ that works for $x\in V(F)\setminus T$ as follows.
\begin{itemize} 
\item Call $\link(x)$ and set $x\gets \find(x)$.
\item In addition, if $\minimal(x)=\varnothing$, then set $\minimal(x)\coloneqq B$.
\end{itemize}
When we detect an interior edge $e$, we update the values of $\minimal$ and the data structure for $\B^*$ by the following procedure $\update(e)$, which uses the operation $\ascend$.
The procedure $\reconst$ in Step~4 will be described later.
Observe that, if Step~2 ends with $x=y$, then the vertex $x$ is the lowest common ancestor of $v$ and $u$ or the calyx of the blossom containing it.
\makeatletter
\renewcommand{\ALG@name}{Procedure}

\begin{algorithm}[h]\caption{~$\update(e)$}
\begin{enumerate}
\item Create a new name $B$ of a blossom. Set $e_B\coloneqq e$.
Let $v$ and $u$ be the vertices with $\partial e=\{v,u\}$. If $\minimal(v)=\varnothing$, then $\minimal(v)\coloneqq B$.
If $\minimal(u)=\varnothing$, then $\minimal(u)\coloneqq B$.
\item Set $x\gets \find(v)$ and $y\gets \find(u)$. While $x\neq y$ and $\{x,y\}\not\subseteq T$, do the following:
\begin{itemize}
\item If $\seq(x)>\seq(y)$, then call $\ascend(x)$.
\item If $\seq(x)<\seq(y)$, then call $\ascend(y)$. 
\end{itemize}
\item If $x=y\not\in T$, then set $z\gets x$. While the stalk of $z$ is labeled, call $\ascend(z)$.  
Set $w_B\coloneqq z$.
\item Otherwise (i.e., when $\{x,y\}\subseteq T$),  reconstruct a short augmenting walk by $\reconst(e,x,y)$.
\end{enumerate}
\end{algorithm}
We now analyze the total numbers of queries for $\link$ and $\find$ through $\searchfast$.
Since $\link$ is applied to the representative elements of members of $\B^*$ and they are no more representative elements afterwards, $\link$ is called $O(|V|)$ times.
Note that the number of queries of $\find$ in the creation of one blossom is at most two larger than that of $\link$. 
As $\B$ is a laminar family, the number of blossoms is at most $2|V|$.  Thus, $\find$ is called $O(|V|)$ times in total in creations of blossoms.
We remark that $\find$ is also used to check whether an edge is interior or not because the existence of $B\in\B$ with $\{v,u\}\subseteq B$ 
is equivalent to $\find(v)=\find(u)$. The number of queries for this purpose is $O(|E|)$ 
because $\scan$ is called $O(|E|)$ times in $\searchfast$ by Lemma~\ref{lem:scan_linear}.

\paragraph{Reconstruction of the output walk.}
Since our implementation does not store search walks explicitly, we need to reconstruct the walk $Q\coloneqq W_i(v)\cdot e\cdot \overline{W_j(u)}$ when $\searchfast$ terminates at Step~2.
Before describing a procedure for reconstruction, we observe the following property of $Q$.   
For any blossom $B\in \B$ that intersects $Q$, all appearances of vertices in $B$ forms one subwalk in $Q$, and this subwalk starts or ends at the calyx $w_B$. 
This follows from the definitions of search walks in $\scan$ and the property ($\star$) in Lemma~\ref{lem:shrink}.

The procedure $\reconst$ described below has a recursive form.
If $\reconst(e,x,y)$ is called in $\update(e)$ for an edge $e$ with $\partial e=\{v,u\}$, then it returns $Q= W_i(v)\cdot e\cdot \overline{W_j(u)}$ or its reverse for some indices $i$ and $j$, where $x$ and $y$ are the first and last vertices of the returned walk, respectively. In this case, the walks $Q_v$ and $Q_u$ constructed in the procedure are $\overline{W_i(v)}$ and $\overline{W_j(u)}$, respectively.
If~$\reconst$ is called recursively, it is used in the form of $\reconst(e_{B}, v, w_{B})$ for the edge $e_B$ associated to some blossom $B$, the calyx $w_B$, and some $v\in B$.
The walk returned by $\reconst(e_{B}, v, w_{B})$ corresponds to the subwalk of $Q$ consisting of all vertices of $B$ appearing in $Q$.
For simplicity, we use the dummy symbol $*$, which satisfies $*\neq t$ for all terminals $t\in T$.

\begin{algorithm}[h]\caption{~$\reconst(e,x,y)$}
\begin{enumerate}
\item Let $v$ and $u$ be the vertices with $\partial e=\{v,u\}$. Set walks $Q_v\gets v$ and $Q_u\gets u$. 
Set $(s,t)$ by
$$(s,t)\gets \left\{\begin{array}{ll}
(\sigma_{E^*}(e,v), \sigma_{E^*}(e,u)) & (\text{$e$ is labeled}) \\
(*, *) & (\text{$e$ is free and $\sym(v)\neq \sym(u)$}) \\
(\sym(v), *) & (\text{$e$ is free,  $\sym(v)= \sym(u)$, and $u\not\in V(\B)$})\\
(*, \sym(u)) & (\text{otherwise}).
\end{array}\right.$$
\item While $v \not\in \{x,y\}$, do the following: 
\begin{enumerate}
\item Let $\tilde{e}$ and $\tilde{v}$ be the stalk and parent of $v$, respectively.
\item If $s \neq \sym(v)$ or $\tilde{e}$ is an edge detected as exterior in $\searchfast$, then  $Q_v\gets Q_v\cdot \tilde{e}\cdot \tilde{v}$ and $v\gets \tilde{v}$. 
In particular, in case $\tilde{e}$ is labeled, set $s\gets \sigma_{E^*}(\tilde{e},\tilde{v})$.
\item Otherwise, set $B'\gets \minimal(v)$,
let $Q'$ be the walk returned by $\reconst(e_{B'}, v, w_{B'})$,
and set $Q_v\gets Q_v+Q'$, $v\gets w_{B'}$, and $s\gets *$.
\end{enumerate}
\item
While $\{v, u\}\neq \{x,y\}$, do (a)--(c) with $u$, $Q_u$, and $t$ in places of $v$, $Q_v$, and $s$, respectively.

\item If $(v,u)=(x,y)$, then return $Q\coloneqq \overline{Q_v}\cdot e\cdot Q_u$.  
Otherwise (i.e., if $(v,u)=(y, x)$), return $\overline{Q}$.  
\end{enumerate}
\end{algorithm}
When called in $\update(e)$, the procedure $\reconst(e,x,y)$ returns the walk $W_i(v)\cdot e\cdot \overline{W_j(u)}$ in time linear to the length of the returned walk, which is $O(|V|)$ by Lemma~\ref{lem:C1}.
It takes $O(|E|)$ time to transform the walk $W_i(v)\cdot e\cdot \overline{W_j(u)}$ in $\G^*(\P)$ to the corresponding short augmenting walk in the original labeled graph $\G(\P)$. Thus the additional running time required to obtain a short augmenting walk in $\G(\P)$ after $\searchfast$ is $O(|E|)$.

\bigskip
By combining the implementation details explained above with  Lemmas~\ref{lem:scan-linear-correct} and \ref{lem:scan_linear},
we obtain the following theorem.

\begin{theorem}
The algorithm $\searchfast$ runs in $O(|E|)$ time and outputs a short augmenting walk in $\G^*(\P)$ if it admits one and returns a message``no augmenting walk'' if $\P$ consists of the maximum number of edge-disjoint $T$-paths.
\end{theorem}

We now summarize the above results to analyze the overall running time of our algorithm for finding maximum number of edge-disjoint $T$-paths. Given a collection $\P$ of edge-disjoint $T$-paths, $\searchfast$ finds a short augmenting walk $Q^*$ in $\G^*(\P)$ or certify the nonexistence of an augmenting walk in $O(|E|)$ time. 
Then one can transform $Q^*$ to a short augmenting walk $Q$ in $\G(\P)$ in $O(|E|)$ time. It also takes $O(|E|)$ time to obtain $\P\triangle Q$ by the switching operation. Thus the running time of our algorithm is 
$O(|E|)$ per augmentation. Since the number of augmentation is bounded by the maximum number of edge-disjoint $T$-paths, which is $O(|E|)$, we may conclude that our algorithm runs in $O(|E|^2)$ time. 

\begin{theorem}
Given a multigraph $G=(V,E)$ and a set $T\subseteq V$ of terminals, one can find maximum number of edge-disjoint $T$-paths in $O(|E|^2)$ time.
\end{theorem}

\section{The Edmonds--Gallai Type Decomposition}\label{sec:EG}
Analogously to the Edmonds--Gallai decomposition for maximum matching, Seb\H{o} and Szeg\H{o}~\cite{SS04}
introduced a canonical decomposition for maximum edge-disjoint $T$-paths as follows. 
Let $k$ be the maximum number of edge-disjoint $T$-paths in $G$. For a terminal $s\in T$, 
a vertex $u\in V$ is said to be {\em $s$-rooted} if there exists a family $\P$ of $k$ edge-disjoint $T$-paths 
and an $s$-$u$ path that is edge-disjoint from all the members in $\P$. 
A terminal $s\in T$ is $s$-rooted by definition, and not $t$-rooted for any other terminal $t\in T\setminus\{s\}$.  

We say that a $T$-subpartition $\Y$ is {\em optimal} if it minimizes $\kappa(\Y)$ among all the $T$-subpartitions,
where $\kappa(\Y)$ is defined in Theorem~\ref{th:Mader}.
The following observation is immediate therefrom. 

\begin{lemma}[{\cite[Theorem 5]{SS04}}]
\label{lem:EGD}
For any optimal $T$-subpartition $\Y=\{Y_s\}_{s\in T}$, let $Y_0$ denote the union of the vertex sets of 
even components in $G\setminus\Y$. Then an arbitrary vertex $v\in Y_0$ is not $s$-rooted for any $s\in T$. In addition, 
an arbitrary vertex $v\in Y_t$ is not $s$-rooted for any $s\in T\setminus\{t\}$.  
\end{lemma}

Let $V_s$ denote the set of vertices that are $s$-rooted and not $t$-rooted for any other $t\in T$. 
Then each $s\in T$ is $s$-rooted by definition, and the family $\V\coloneqq \{V_s\}_{s\in T}$ forms a $T$-subpartition. 
We denote by $V_0$ the set of vertices that are not $s$-rooted for any $s\in T$. 
A vertex $v\in V$ is said to be {\em double-rooted} if it is $s$-rooted and $t$-rooted for any pair of distinct terminals 
$s,t\in T$. The set of double-rooted vertices is denoted by $V_\infty$. 

Seb\H{o} and Szeg\H{o}~\cite{SS04} showed that the $T$-subpartition $\V$ provides a canonical decomposition that is 
analogous to the Edmonds--Gallai decomposition.  
In this section, we show that the $T$-subpartition obtained at the end of our algorithm coincides with $\V$, 
which means that our algorithm finds not only maximum edge-disjoint $T$-paths but also the Edmonds--Gallai type decomposition. 

Suppose that $\P$ consists of maximum edge-disjoint $T$-paths. Then $\search$ applied to $\G(\P)$ terminates 
without finding an augmenting walk. Define a $T$-subpartition $\X=\{X_s\}_{s\in T}$ as in the proof of Theorem~\ref{thm:tightness}.
That is, $X_s\coloneqq \set{v\in V| v\in V(F)\setminus V(\B),~ \lambda(W_1(v))=s}$ for each $s\in T$.
In addition, set $X_\infty\coloneqq V(\B)$ and $X_0\coloneqq V\setminus V(F)$.
Then, $\X$ is an optimal $T$-subpartition by Theorem~\ref{thm:tightness}.
The proof of this theorem also implies that  $X_\infty$ is the union of the vertex sets of odd components in $G\setminus\X$
and $X_0$ is the union of the vertex sets of even components in $G\setminus\X$. 
\begin{lemma}
\label{lem:rooted}
The $T$-subpartition $\X$ and the vertex subset $X_\infty$ satisfy the following. 
\begin{itemize}
\item For each terminal $s\in T$, every vertex in $X_s$ is $s$-rooted. 
\item Every vertex in $X_\infty$ is double-rooted. 
\end{itemize}
\end{lemma}
\begin{proof}
Take any terminal $s\in T$ and any vertex  $u\in X_s$.
For an arbitrary terminal $r\in T$ with $r\neq s$, consider a graph $G_{ur}=(V,E\cup\{f\})$ obtained by attaching a new edge $f$ with $\partial f=\{u,r\}$ to $G$. 
In addition, let $\G_{ur}(\P)$ be obtained from $\G(\P)$ by attaching a free edge $f$ to $\G(\P)$.
By Lemmas~\ref{lem:C2}, \ref{lem:no-shortcut}, and $\lambda(W_1(u))=s\neq r$,
the walk $W_1(u)\cdot f\cdot r$ determines a short augmenting walk in $\G_{ur}(\P)$, which we denote by $Q_{ur}$. 
Then,  Theorem~\ref{thm:SAW1} implies that $\P\triangle Q_{ur}$ gives $k+1$ edge-disjoint $T$-paths. 
Among them, let $P$ be the one that contains $f$.
Then, the end-terminals of $P$ are $r$ and some $t\in T$ with $t\neq r$.
If $t=s$, deleting $f$ and $r$ from $P$ yields an $s$-$u$ path that is edge-disjoint from the other $k$ edge-disjoint $T$-paths.
We now prove $t=s$ to complete  the first claim.

Suppose to the contrary that $t\neq s$.
Consider a graph $G_{ut}=(V, E\cup \{f'\})$ with a new edge $f'$ with $\partial f'=\{u,t\}$ and let $\G_{ut}(\P)$ be obtained from $\G(\P)$ by attaching $f'$.
Since $\lambda(W_1(u))=s\neq t$, we can similarly show that $W_1(u)\cdot f'\cdot t$ determines a short augmenting walk $Q_{ut}$,
and $\P\triangle Q_{ut}$ yields $k+1$ edge-disjoint $T$-paths. 
Let $P'$ be the one that contains $f'$.
Since $Q_{ur}$ and $Q_{ut}$ differ only at the last edges $f,f'$ and vertices $r,t$, the definition of switching operation implies that
$P'$ is obtained from $P$ by replacing $f$ and $r$ with $f'$ and $t$, respectively, which implies that the end-terminals of $P'$ are both $t$, a contradiction.    

To show the second claim, take any  $u\in X_\infty$ and let $s\coloneqq \lambda(W_1(u))$.
By applying the argument above, we can obtain that $u$ is $s$-rooted.
Next, consider a graph $G_{us}=(V,E\cup\{f''\})$, where $f''$ is a new edge with $\partial f''=\{u,s\}$.
Then, the walk $R_{us}\coloneqq W_2(u)\cdot f''\cdot s$ determines a short augmenting walk in $\G_{us}(\P)$ by
Lemmas~\ref{lem:C2}, \ref{lem:no-shortcut}, and $\lambda(W_2(u))\neq s$.
Then, Theorem~\ref{thm:SAW1} implies that $\P\triangle R_{us}$ yields $k+1$ edge-disjoint $T$-paths. 
Let $P''$ be the one that contains $f''$.
Then the end-terminals of $P''$ are $s$ and some $t\in T$ with $t\neq s$.
Deleting $f''$ and $s$ from $P''$ yields a $t$-$u$ path that is edge-disjoint from the other $k$ edge-disjoint $T$-paths.
Therefore, $u$ is $t$-rooted. Since $s\neq t$, we may conclude that $u$ is double-rooted. 
\end{proof}

Combining Lemmas~\ref{lem:EGD} and \ref{lem:rooted} with Theorem~\ref{thm:tightness}, we obtain the following theorem, 
which shows that the Edmonds--Gallai type decomposition is obtained at the end of our algorithm for maximum edge-disjoint 
$T$-paths.   
\begin{theorem}
\label{th:EGD}
The $T$-subpartition $\X$ provides the Edmonds--Gallai type decomposition, 
namely $\X=\V$, $V_0=X_0$, and $V_\infty=X_\infty$ hold.  
\end{theorem}
\begin{proof}
Since $\X$ is an optimal $T$-subpartition by Theorem~\ref{thm:tightness}, it follows from Lemma~\ref{lem:EGD} that 
the vertices in $X_s$ are not $t$-rooted for any $t\in T\setminus\{s\}$, which together with Lemma~\ref{lem:rooted} implies 
that $X_s\subseteq V_s$ holds. 
Lemma~\ref{lem:EGD} and Theorem~\ref{thm:tightness} also imply that a vertex in $X_0$ is not $s$-rooted for any $s\in T$.
Hence we have $X_0\subseteq V_0$. By Lemma~\ref{lem:rooted}, each vertex in $X_\infty$ is double-rooted, which implies that 
$X_\infty\subseteq V_\infty$ holds. Since the union of $X_0$,  $X_\infty$, and  $X_s$ for all $s\in T$ coincides with the 
entire vertex set $V$, we have $\X=\V$, $X_0=V_0$, and $X_\infty=V_\infty$.  
\end{proof}

This theorem and the proof of Theorem~\ref{thm:tightness} immediately provide an alternative proof to the structure theorem for edge-disjoint $T$-paths by 
Seb\H{o} and Szeg\H{o}~\cite{SS04}. 
\begin{corollary}[{\cite[Theorem 7]{SS04}}]
\label{cor:EGD}
The family $\V=\{V_s\}_{s\in T}$ is an optimal $T$-subpartition. The set $V_0$ coincides with the union of 
the vertex sets of the even components of $G\setminus \V$, and $V_\infty$ is the union of the vertex sets of the 
odd components of $G\setminus\V$. 
\end{corollary}

\section{Extension to Integer Free Multiflow}\label{sec:IFMF}
In this section, we present a strongly polynomial algorithm for the integer free multiflow problem, 
which is  a capacitated version of the maximum edge-disjoint $T$-paths problem.  

Let $G=(V,E)$ be a simple undirected graph with a set $T\subseteq V$ of terminals and an edge-capacity function $c:E\to\Zp$. 
A {\em multiflow} $\F$ in the network $N=(G,T,c)$ is represented as a pair  $\F=(\braket{P_1,\dots,P_k},\braket{\alpha_1,\dots\alpha_k})$ 
of $T$-paths $P_1,\dots,P_k$ in $G$ and positive real coefficients $\alpha_1,\dots\alpha_k$ such that
$$\zeta_{\F}(e)\coloneqq \sum\set{\alpha_i\mid e\in E(P_i)}\leq c(e),\quad\quad \forall e\in E,$$
where $E(P_i)$ denote the edge set of $P_i$ for $i=1,\ldots,k$. 

A multiflow is {\em integral} if all coefficients $\alpha_i$ in $\F$ are integers.
The total value of $\F$ is defined by $\val(\F)\coloneqq \alpha_1+\cdots+\alpha_k$.
The {\em free multiflow problem} is to find a multiflow with the maximum total value.
This setting is called {\em free} because no pair of terminals is forbidden to use as ends of a $T$-path.
The {\em integer free multiflow problem} is a variant in which the maximum is taken over all integral multiflows.
We write $\OPT(N)\coloneqq \max\set{\val(\F)|\F:\text{multiflow in $N$}}$ and $\OPTint(N)\coloneqq \max\set{\val(\F)|\F:\text{integral multiflow in $N$}}$
for those maximum values.

Note that 
we can reduce the integer free multiflow problem 
to the maximum edge-disjoint $T$-paths problem by replacing each edge $e$ by $c(e)$ parallel edges.
Then Theorem~\ref{th:Mader} implies the following statement. 
Here, $d_c(X)$ is defined by $d_c(X)\coloneqq \sum\set{c(e)|e\in \delta(X)}$ and $\odd_c(G\setminus\X)$
denotes the number of connected components $K$ in $G\setminus \X$ with odd $d_c(K)$.
\begin{theorem}\label{them:Mader-flow}
For a network $N=(G,T,c)$, the optimal value $\OPTint(N)$ equals the minimum of 
$\kappa_{c}(\X)\coloneqq \frac{1}{2}\left[\sum_{s\in T}d_c(X_s)-\odd_c(G\setminus\X)\right]$
among all the $T$-subpartitions $\X=\{X_s\}_{s\in T}$.
\end{theorem}
Moreover, combined with the search algorithm and the augmenting operation described in Sections~\ref{sec:augmentation} and \ref{sec:search}, 
this reduction implies a pseudopolynomial-time algorithm for 
the integer free multiflow problem. 

If the network $N$ is {\em inner Eulerian}, i.e., if $d_c(v)$ is even for every inner vertex $v\in V\setminus T$, 
we have $\OPT(N)=\OPTint(N)$ by the result of Cherkassky \cite{Cherkassky77} and Lov\'asz \cite{Lovasz76}.
For this special case, a strongly polynomial algorithm is given by Ibaraki, Karzanov, and Nagamochi \cite{IKN98}.
\begin{theorem}[Ibaraki et al. \cite{IKN98}]\label{thm:IKN}
For any inner Eulerian network, one can find an integral free multiflow of maximum total value 
in $O(\varphi(|V|,|E|)\log |T|)$ time, where $\varphi$ is the complexity of finding a maximum flow between two terminals.
The output multiflow can be represented by $O(|E|\log|T|)$ $T$-paths along with integer coefficients.
\end{theorem}
As the current best time complexity for the maximum flow problem is $O(|V|\cdot|E|)$ due to \cite{Orlin13},
their algorithm runs in $O(|V|\cdot|E|\log |T|)$ time. While the algorithm returns an optimal multiflow 
as a collection of flows in the ``node-arc'' form, one can transform it into the path packing form with $O(|E|\log|T|)$ $T$-paths in
$O(|V|\cdot|E|\log|T|)$ time (see the last paragraph of \cite[Section 2]{IKN98}).

In the rest of this section, we provide a strongly polynomial algorithm that solves the integer free multiflow problem 
for general networks. 
Our algorithm starts with finding an optimal fractional multiflow, which is half-integral, and rounding it down 
to obtain an initial integral multiflow. The algorithm then modifies the integral multiflow by repeatedly applying 
the search and augmentation procedures until it certifies that no further augmentation is possible. 
This technique to exploit the half-integral optimal solution was introduced by Pap~\cite{Pap07} for 
the node-capacitated multiflow problem. 
In the following, we describe the details of each phase and complexity analysis of our algorithm.

\paragraph{Finding an Initial Integral Multiflow}\label{sec:initial-flow}

For a given network $N=(G,T,c)$, which may not be inner Eulerian,
{consider the inner Eulerian network $N'\coloneqq (G,T,2c)$ with each edge-capacity being doubled.} 
By Theorem~\ref{thm:IKN}, we can find a maximum multiflow
$\F'=(\braket{P_1,\dots,P_k},\braket{\beta_1,\dots\beta_k})$ such that $\beta_i$ are all integers and $k=O(|E|\log|T|)$.
Then $\F^*\coloneqq (\braket{P_1,\dots,P_k},\braket{\frac{1}{2}\beta_1,\dots, \frac{1}{2}\beta_k})$ is a fractional multiflow 
in $N$ and satisfies $\val(\F^*)=\OPT(N)\geq \OPTint(N)$.

Furthermore, we define $\F\coloneqq (\braket{P_1,\dots,P_k},\braket{\lfloor\frac{1}{2}\beta_1\rfloor,\dots, \lfloor\frac{1}{2}\beta_k\rfloor})$, 
which is an integral multiflow in $N$. Since $k=O(|E|\log|T|)$, we have $\val(\F^*)-\val(\F)=O(|E|\log|T|)$, 
which together with $\val(\F^*)\geq\OPTint(N)$ implies that $\OPTint(N)-\val(\F)=O(|E|\log|T|)$.

Starting with this integral multiflow $\F$, we repeatedly augment by short augmenting walks until 
we certify that no further augmentation is possible in the network. Since $\OPTint(N)-\val(\F)=O(|E|\log|T|)$ holds, 
{the number of iterations is $O(|E|\log|T|)$.}

\paragraph{Augmentation}
Let $\F=(\braket{P_1,\dots,P_k},\braket{\alpha_1,\dots \alpha_k})$ be an integral multiflow in a network $N=(G=(V,E),T,c)$.
We now define an auxiliary labeled graph $\G(\F)=((V,\tilde{E}\cup L), \sigma_V,\sigma_{\tilde{E}},\sigma_L)$. 
{Set} $\P\coloneqq \{P_1,\dots,P_k\}$ and let $E(\P)$ be the edge set defined by the disjoint union of the $T$-paths $P_1, \dots, P_k$.
The set $\tilde{E}$ is obtained by extending $E$ as follows. 
For any $e=\{u,v\}\in E$ with $\zeta_{\F}(e)\leq c(e)-1$, we add one copy of $e$ to $\tilde{E}$ if the equality holds,
and add two parallel copies of $e$ to $\tilde{E}$ if the strict inequality holds.
The labeled graph $\G(\F)$ is defined as in Section~\ref{sec:augmentation} with respect to the graph $G_\F\coloneqq (V,\tilde{E})$ and $\P$.
Then, $E(\P)$ consists of labeled edges and $\tilde{E}\setminus E(\P)$ consists of free edges.
A short augmenting walk in $\G(\F)$ is defined as in Section~\ref{sec:augmentation}.

If the algorithm $\search$ in Section~\ref{sec:search} returns a walk $Q$ for the input $\G(\F)$, then $Q$ is a short augmenting walk.
By Theorem~\ref{thm:SAW1},  the switching operation of $\P$ by $Q$ yields $k+1$ edge-disjoint $T$-paths $\P'=\{P'_1,\dots,P'_{k+1}\}$ on $G(\F)$.
Note that the number of $\P$-segments in $Q$ is $O(|V|)$ by Lemma~\ref{lem:C1} and that a $T$-path in $P\in \P$ is not affected by the switching operation if $Q$ has no $P$-segment.
This implies that the number of $T$-paths in $\P'\setminus \P$ is $O(|V|)$. 
{We now} denote the $T$-paths in $\P'\setminus \P$ by $P'_1, \dots, P'_{\ell}$, where $\ell=O(|V|)$, 
{and define} $\tilde{\F}$ by $\tilde{\F}=(\braket{P_1, \dots,P_k, P'_1, \dots,P'_{\ell}},\braket{\alpha'_{1}, \dots \alpha'_{k}, 1,\dots,1})$,
where $\alpha'_i=\alpha_i$ if $P_i\in \P\cap \P'$ and $\alpha'_i=\alpha_i-1$ otherwise.
By the construction of $G_\F$, we see that $\tilde{\F}$ is an integral multiflow satisfying $\val(\tilde{\F})=\val(\F)+1$.

\paragraph{Correctness}
As explained above,  we can increase the total value of an integral multiflow if $\search$ returns a walk for the input $\G(\F)$ defined for the current multiflow $\F$.
We now show that, if $\search$ applied to $\G(\F)$ terminates without returning a walk, 
then $\F=(\braket{P_1,\dots,P_k},\braket{\alpha_1,\dotsm\alpha_k})$ is a maximum integral multiflow in $N$, i.e., $\val(\F)=\OPTint(N)$. 

Let $F$ and $\B$ be the forest and the laminar family at the termination of $\search$ applied to $\G(\F)$.
Define a $T$-subpartition $\X=\{X_s\}_{s\in T}$ in the same way as in 
the proof of Theorem~\ref{thm:tightness}. 
By the same arguments, we can verify the three claims presented there. We now use them to show that $\val(\F)=\frac{1}{2}\left[\sum_{s\in T}d_c(X_s)-\odd_c(G\setminus\X)\right]$,
which guarantees the maximality of $\F$ by Theorem~\ref{them:Mader-flow}.
By the second claim, every $T$-path $P_i$ contributes $2\alpha_i$ to $\sum_{s\in T}d_{c}(X_s)$.
The first claim implies $\zeta_{\F}(e)=c(e)$ for every edge connecting $X_s$ and $X_t$ with distinct $s,t\in T$.
The third claim implies that, among the edges connecting $\bigcup_{s\in T}X_s$ and any connected component $K$ of $G\setminus \X$, at most one edge satisfies $\zeta_{\F}(e)=c(e)-1$ and other edges satisfy $\zeta_{\F}(e)=c(e)$.
As $\sum\set{\zeta_{\F}(e)|e\in \delta(K)}$ should be even, 
there is exactly one free edge between $\bigcup_{s\in T}X_s$ and $K$ if $d_c(K)$ is odd and otherwise there is no such free edge. 
Therefore, we obtain $\sum_{s\in T}d(X_s)=2\val(\F)+\odd_c(G\setminus\X)$, which is equivalent to the required equation.

\paragraph{Time Complexity}
We now analyze the time complexity of our algorithm.
As mentioned above, finding an initial integral flow can be done in $O(|V|\cdot|E|\log |T|)$ time by using the algorithm of Ibaraki, et al.~\cite{IKN98}.
As this initial flow $\F$ satisfies $\OPTint(N)-\val(\F)=O(|E|\log|T|)$, 
the number of augmentations is $O(|E|\log|T|)$.
At first, $\F$ consists of $O(|E|\log|T|)$ $T$-paths.
Each augmentation may increase the number of $T$-paths used in $\F$ by $O(|V|)$ as mentioned above.
{Hence, throughout} the algorithm, the number of $T$-paths constituting $\F$ is $O(|V|\cdot|E|\log|T|)$.
Since each path has length at most $|V|$, the size of the labeled graph $\G(\F)$ is $O(|V|^2\cdot|E|\log|T|)$.
The search and augmentation are applied $O(|E|\log|T|)$ times,  
and each application runs in linear time. Therefore,
the total running time of the algorithm is $O(|V|^2\cdot|E|^2\log^2|T|)$.

\section*{Acknowledgements}
The authors thank Hiroshi Hirai for his suggestion on the maximum integer free multiflow problem, Etsuko Ishii for her help in reading the German literature, and Andr\'as Seb\H{o} for informing us of his related works. We are also grateful to the anonymous reviewers 
for helpful suggestions.

\end{document}